\documentclass[11pt]{article}
\usepackage[usenames,dvipsnames,svgnames,table]{xcolor} 

\usepackage{enumitem} 
\usepackage{rotfloat} 
\usepackage{graphicx}
\usepackage{epsfig}
\usepackage{amssymb, amsmath}
\usepackage{verbatim}
\usepackage{natbib}
\usepackage{authblk}
\usepackage{enumitem}
\usepackage{multirow}
\usepackage{booktabs} 
\usepackage{adjustbox}
\usepackage[colorlinks=true,linkcolor=black,citecolor=black,urlcolor=black]{hyperref}%
\usepackage[colorinlistoftodos, textsize=scriptsize]{todonotes} 
\usepackage{subcaption} 
\usepackage{algorithm}
\usepackage{algpseudocode}
\usepackage[toc,page]{appendix}

\newtheorem{theorem}{Theorem}
\newtheorem{lemma}[theorem]{Lemma}

\newtheorem{case}{Case}

\newenvironment{proof}[1][Proof]{\begin{trivlist}
\item[\hskip \labelsep {\bfseries #1}]}{\end{trivlist}}
\newtheorem{definition}{Definition}

\newcommand{\ech}{\color{black}  ~}    

\newcommand{\bea}{\begin{eqnarray*}}
\newcommand{\eea}{\end{eqnarray*}}
\newcommand{\bean}{\begin{eqnarray}}
\newcommand{\eean}{\end{eqnarray}}

\newcommand{\bxi}{\boldsymbol{\xi}}

\newcommand{\calX}{\mathcal{X}}

\begin{document}

\title{L\'{e}vy Adaptive B-spline Regression
via Overcomplete Systems}
\author[1]{Sewon Park}
\author[1]{Hee-Seok Oh}
\author[1]{Jaeyong Lee}
\affil[1]{Department of Statistics, Seoul National University}

\maketitle

\begin{abstract}
The estimation of functions with varying degrees of smoothness is a challenging problem in the nonparametric function estimation.  In this paper, we propose the LABS (L\'{e}vy Adaptive B-Spline regression) model, an extension of the LARK models, for the estimation of functions with varying degrees of smoothness. LABS model is a LARK with B-spline bases as generating kernels.  The B-spline basis consists of piecewise $k$ degree polynomials with $k-1$ continuous derivatives and can express systematically functions with varying degrees of smoothness. By changing the orders of the B-spline basis, LABS can systematically adapt the smoothness of functions, i.e., jump discontinuities, sharp peaks, etc. Results of simulation studies and real data examples support that this model catches not only smooth areas but also jumps and sharp peaks of functions. The proposed model also has the best performance in almost all examples. Finally, we provide theoretical results that the mean function for the LABS model belongs to the certain Besov spaces based on the orders of the B-spline basis and that the prior of the model has the full support on the  Besov spaces.

\bigskip
\noindent Key words: Nonparametric Function Estimation;  L\'{e}vy Random Measure;  Besov Space;  Reversible Jump Markov Chain Monte Carlo.
\end{abstract}

\newpage


\section{Introduction}

Suppose we observe $n$ pairs of observations, $(x_1, y_1),(x_2, y_2), \ldots, (x_n, y_n)$ where
\begin{equation}
y_i = \eta(x_i) + \epsilon_i, \quad \epsilon_i \stackrel{iid}{\sim}\mathcal{N}(0, \sigma^2),\quad i  = 1,2, \ldots, n,
\label{eq:funcrelation}
\end{equation}
and $\eta$ is an unknown real-valued function which maps  $\mathbb{R}$  to $\mathbb{R}$. We wish to estimate the mean function $\eta$ that may have varying degrees of smoothness including discontinuities.  In nonparametric function estimation, we often face smooth curves except for a finite number of jump discontinuities and sharp peaks, which are common in many climate and economic datasets.  Heavy rainfalls cause a sudden rise in the water level of a river.  The COVID-19 epidemic brought about a sharp drop in unemployment rates. Policymakers' decisions can give rise to abrupt changes. For instance, the United States Congress passed the National Minimum Drinking Age Act in 1984, which has been debated over several decades  in the United States, establishing 21 as the minimum legal alcohol purchase age. This act caused a sudden rise in mortality for young Americans around 21. The abrupt changes can provide us with meaningful information on these issues, and it is important to grasp the changes.

There has been much research into the estimation of local smoothness of the functions. The first approach is to minimize the penalized sum of squares based on a locally varying smoothing parameter or penalty function across the whole domain. \cite{pintore2006spatially}, \cite{liu2010data}, and \cite{wang2013smoothing} modeled the smoothing parameter of smoothing spline to vary over the domain. \cite{ruppert2000theory}, \cite{crainiceanu2007spatially}, \cite{krivobokova2008fast}, and \cite{yang2017adaptive} suggested the penalized splines based on the local penalty that adapts to spatial heterogeneity in the regression function. The second approach is the adaptive free-knot splines that choose the number and location of the knots from the data. \cite{friedman1991multivariate} and \cite{luo1997hybrid} 
determined a set of knots using stepwise forward/backward knot selection procedures. \cite{zhou2001spatially} avoided the problems of stepwise schemes and proposed optimal knot selection schemes introducing the knot relocation step. \cite{smith1996nonparametric}, \cite{denison1998automatic}, \cite{denison1998bayesian}, and \cite{dimatteo2001bayesian}
studied Bayesian estimation of free knot splines using MCMC techniques. The third approach is to use wavelet shrinkage estimators including VisuShrink based on the universal threshold \citep{donoho1994ideal}, SureShrink based on Stein's unbiased risk estimator (SURE) function \citep{donoho1995adapting}, Bayesian thresholding rules by utilizing a mixture prior \citep{abramovich1998wavelet}, and empirical Bayes methods for level-dependent threshold selection \citep{johnstone2005empirical}. The fourth approach is to detect jump discontinuities in the regression curve. \cite{koo1997spline}, \cite{lee2002automatic}, and \cite{yang2014jump} dealt with the estimation of discontinuous function using B-spline basis functions. \cite{qiu1998local}, \cite{qiu2003jump}, \cite{gijbels2007jump}, and \cite{xia2015jump} identified jumps based on local polynomial kernel estimation.

In this paper, we consider a function estimation method using overcomplete systems.  A subset of the vectors $\{\phi\}_{j \in J}$  of Banach space $\mathcal{F}$ is called a {\it complete system}  if
\begin{equation*}
\| \eta - \sum_{j \in J} \beta_j \phi_j\| < \epsilon, \quad   \forall \eta \in \mathcal{F},  \forall \epsilon > 0,
\end{equation*}
where $\beta_j \in \mathbb{R}$ and $J \in \mathbb{N} \cup \infty$. Such a complete system is {\it overcomplete} if removal of a vector $\phi_j$ from the system does not alter the completeness. In other words, an overcomplete system is constructed by adding basis functions to a complete basis \citep{lewicki2000learning}. Coefficients $\beta_j$ in the expansion of $\eta$ with an overcomplete system are not unique owing to the redundancy intrinsic in the overcomplete system. The non-uniqueness property can provide more parsimonious representations than those with a complete system \citep{simoncelli1992shiftable}.

 The L\'{e}vy Adaptive Regression Kernels (LARK) model, first proposed by \cite{tu2006bayesian}, is a Bayesian regression model  utilizing overcomplete systems with L\'{e}vy process priors.  \cite{tu2006bayesian} showed the LARK model had sparse representations for $\eta$ from an overcomplete system and improvements in nonparametric function estimation.  \cite{pillai2007characterizing} found out the relationship between the LARK model and a reproducing kernel Hilbert space (RKHS), and \cite{pillai2008levy}  proved the posterior consistency of the LARK model. \cite{ouyang2008bayesian} extended the LARK method to the classification problems. \cite{chu2009bayesian} used continuous wavelets as the elements of an overcomplete system. \cite{wolpert2011stochastic} obtained sufficient conditions for LARK models to lie in the some Besov space or Sobolev space.  \cite{lee2020bayesian} devised an extended LARK model with multiple kernels instead of only one type kernel.

In this paper, we develop a fully Bayesian approach with B-spline basis functions as the elements of an overcomplete system and call it the L\'{e}vy Adaptive B-Spline regression (LABS).  Our main contributions of this work can be summarized as follows.
\begin{enumerate}
    \item The LABS model can systematically represent the smoothness of functions  varying  locally by changing the orders of the B-spline basis. The form of a B-spline basis depends on the locations of knots and can be symmetrical or asymmetrical. The varying degree of B-spline basis enables the LABS model to adapt to the smoothness of functions.
    \item We investigate two theoretical properties of the LABS model. First, the mean function of the LABS model exists in certain Besov spaces based on the types of degrees of B-spline basis.  Second, the prior of the LABS model has full support on some Besov spaces. Thus, the proposed LABS model extends the range of smoothness classes of the mean function. 
    \item We provide empirical results demonstrating that our model performs well in the spatially inhomogeneous functions such as the functions with both jump discontinuities, sharp peaks, and smooth parts. The LABS model achieved the best results in almost every experiments compared to  the popular nonparametric function estimation methods. In particular, the LABS model showed remarkable performance in estimating functions with jump discontinuities and outperformed other competing models. 
\end{enumerate}

The rest of the paper is organized as follows.  In \autoref{sec:lark}, we introduce the L\'{e}vy Adaptive Regression Kernels and discuss its properties. In \autoref{sec:newmodel}, we propose the LABS model and present an equivalent model with latent variables  that make the posterior computation tractable. We present three theorems that  the function spaces for the proposed model depend upon the degree of B-spline basis and that  the prior has large support in some function spaces. We describe the detailed algorithm of posterior sampling using reversible jump Markov chain Monte Carlo in \autoref{sec:algorithm}. In \autoref{sec:simulation}, the  LABS model is compared with other methods in  two simulation studies and in \autoref{sec:application} three real-world data sets are analysed using the LABS model. In the last section, we discuss some improvements and possible extensions of the proposed model.


\section{L\'{e}vy adaptive regression kernels}\label{sec:lark}

In this section, we give a brief introduction to the LARK model.  Let $\Omega$ be a complete separable metric space, and $\nu$ be a positive  measure on $\mathbb{R} \times \Omega$ with $ \nu(\{0\}, \Omega) = 0$  satisfying $L_1$ integrability condition,
\begin{equation}
 \label{eq:cond1}
 \int \int_{\mathbb{R} \times A} (1  \wedge |\beta|) \nu(d\beta, d\omega) < \infty,
 \end{equation}
 for each compact set $A \subset \Omega$. The L\'{e}vy random measure $L$ with L\'{e}vy measure $\nu$ is defined as
 \begin{equation*}
 L(d \omega) =  \int_{\mathbb{R}} \beta N(d \beta, d \omega),
 \end{equation*}
 where $ N$ is a Poisson random measure with intensity measure $\nu$. We denote $L \sim \text{L\'{e}vy}(\nu)$. For  any $t \in \mathbb{R}$, the characteristic function of $L(A)$  is
 \vspace{0.1cm}
 \begin{equation}
 \label{eq:LK3}
 \mathbb{E}\left[e^{itL(A)}\right] = \exp\left\{\int \int_{\mathbb{R} \times A} (e ^{it\beta} - 1 ) \nu(d\beta, d\omega)\right\}, \quad \text{for all}\,\, A \subset \Omega.
 \end{equation}

Let $g(x, \omega)$  be a real-valued function defined on $\mathcal{X} \times \Omega$ where $\mathcal{X}$ is another set. By integrating $g$ with respect to a L\'{e}vy random measure $L$, we define a real-valued function on $\mathcal{X}$:
\begin{equation}
\label{eq:PR}
\eta(x) \equiv L[g(x)] = \int_{\Omega} g(x,\omega) L(d \omega) = \int_{\Omega}\int_{\mathbb{R}} g(x,\omega) \beta N(d\beta, d \omega), \forall x \in \mathcal{X}.
\end{equation}
We call $g$ a {\it generating function} of $\eta$.

When $\nu(\mathbb{R} \times \Omega) = M$ is finite,  a L\'{e}vy random measure can be represented as $L(d \omega) = \sum_{j \leq J} \beta_j \delta_{\omega_j}$, where $J$ has a Poisson distribution with mean  $M > 0$ and $\{(\beta_j, \omega_j)\}\stackrel{iid}{\sim}  \pi(d\beta_j, d \omega_j) :=\nu /M, j = 1,2, \ldots, J$. In this case, equation \eqref{eq:PR} can be expressed as
\vspace{-0.3cm}
\begin{equation}
\label{eq:PR2}
\eta(x) =  \sum_{j=1}^{J} g(x,\omega_j) \beta_j,
\end{equation}
 where $\{(\beta_j, \omega_j)\}$ is the random set of finite support points of a Poisson random measure. If $g$ is bounded, $L_1$ integrability condition \eqref{eq:cond1} implies  the existence of \eqref{eq:PR} for all $x$. See \cite{lee2020bayesian}.

 If a L\'{e}vy measure satisfying \eqref{eq:cond1} is infinite, the number of the support points of $N(\mathbb{R},\Omega)$  is infinite almost surely.  \cite{tu2006bayesian} proved that the truncated  L\'{e}vy random field $L_{\epsilon}[g]$ converges in distribution to $L[g]$ as $\epsilon \rightarrow 0$, where
   \begin{equation*}
L_{\epsilon}[g] = \int \int_{[-\epsilon, \epsilon]^{c} \times \Omega} g(x,\omega) \beta N(d\beta, d \omega) = \int \int_{\mathbb{R} \times \Omega} g(x,\omega) \beta N_{\epsilon}(d\beta, d \omega),
 \end{equation*}
 and $N_\epsilon$ is a Poisson measure on $\mathbb{R} \time \Omega$ with mean measure
 \begin{equation*}
 \nu_\epsilon(d\beta,  d \omega) := \nu(d\beta,  d \omega)I_{|\beta| > \epsilon}.
 \end{equation*}
This truncation often used as an approximation of the posterior. For posterior computation methods for the Poisson random measure without truncation, see \cite{lee2007sampling} and \cite{lee2004new}.

Together with data generating mechanism \eqref{eq:funcrelation}, the LARK model is defined as follows:
 \begin{gather*}
\mathbb{E}[Y | L, \theta] = \eta(x) \equiv  \int_{\Omega} g(x,\omega) L(d \omega)\\
L | \theta \sim \text{L\`{e}vy}(\nu) \\
\theta \sim \pi_\theta (d \theta),
\end{gather*}
where  L\`{e}vy($\nu$)  denotes the  L\`{e}vy process which has the characteristic function and $\nu$ is a L\`{e}vy measure satisfying \eqref{eq:cond1}.  \cite{tu2006bayesian} used gamma, symmetric gamma, and symmetric $\alpha$-stable (S$\alpha$S) ($0 < \alpha < 2$)  L\`{e}vy  random fields.  The conditional distribution for $Y$ has a hyperparameter  $\theta$ and $\pi_\theta$ denotes the prior distribution of $\theta$.  The generating function $g(x, \omega)$  is used as elements of an overcomplete system.\cite{tu2006bayesian}  and \cite{lee2020bayesian} used the Gaussian kernel, the Laplace kernel, and Haar wavelet as generating functions:
\begin{itemize}
    \item Haar kernel: $g(x, \omega) := I \left(|\frac{x-\chi}{\lambda}| \leq 1\right)$
    \item Gaussian kernel: $g(x, \omega) = \exp \left\{-\frac{(x-\chi)^2}{2\lambda^2} \right\}$
    \item Laplacian Kernel: $g(x, \omega) = \exp \left\{-\frac{|x-\chi|}{\lambda} \right\}$
\end{itemize}
with $\omega := (\chi, \lambda) \in \mathbb{R} \times \mathbb{R}^{+} := \Omega$.
All of the above generating functions are bounded.

This LARK model can be represented in a hierarchical structure as follows:
\begin{align*}
Y_i |\, \eta(\textbf{x}_i)  &\stackrel{ind}{\sim} \mathcal{N}( \eta(\textbf{x}_i), \sigma^2) \\
 \eta(\textbf{x}_i) &= \sum_{j = 1}^{J}g(\textbf{x}_i, \boldsymbol{\omega}_j) \beta_j \\
 J\, |\, \epsilon &\sim \text{Pois}(\nu_\epsilon(\mathbb{R}, \Omega)) \\
 (\beta_j, \boldsymbol{\omega}_j) | J, \epsilon &\stackrel{i.i.d}{\sim} \pi(d \beta_j, d\boldsymbol{\omega}_j) := \frac{\nu_\epsilon(d \beta_j, d\boldsymbol{\omega}_j) }{\nu_\epsilon(\mathbb{R}, \Omega)}
\end{align*}
 for $j = 1,\ldots, J$. $J$ is the random number that is stochastically determined by L\`{e}vy random measure, $(\beta_1,\ldots, \beta_J)$ is the unknown coefficients of a mean function  and
$(\boldsymbol{\omega}_1,\ldots, \boldsymbol{\omega}_J)$ is the parameters of the generating functions. To obtain samples from the posterior distribution under the LARK model, the reversible jump Markov chain Monte Carlo (RJMCMC)  proposed by \cite{green1995reversible} is used because some parameters have  varying dimensions.

The LARK model stochastically extracts features and finds a compact representation for $\eta(\cdot)$ based on an overcomplete system. That is, it enables functions to be represented by the small number of elements from an overcomplete system. However, both the LARK model and most methods for function estimation use only one type of kernel or basis and can find out the restricted smoothness of the target function. These models cannot afford to capture all parts of the function with various degrees of smoothness. For example, we consider a noisy modified Heavisine function sampled at $n = 512$ equally spaced points on $[0, 1]$ in \autoref{fig:lark_weakness}. The data contains both smooth and non-smooth regions such as peaks and jumps. As shown in panel (a) of \autoref{fig:lark_weakness}, it is difficult for the LARK model with a finite L\`{e}vy measure using Gaussian kernel to estimate jump parts of the data.  We, therefore, propose a new model which can adapt the smoothness of function systematically by using a variety of B-spline bases as the generating elements of an overcomplete system.

\begin{figure}[ht!]
 \centering
 \includegraphics[width=\textwidth]{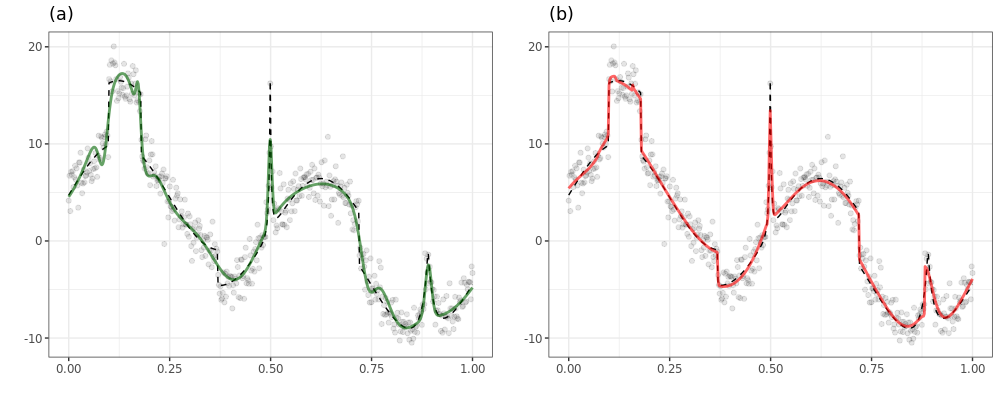}
 \caption{Comparison of curve fitting functions with (a) LARK, and  (b) LABS model for the modified heavisine dataset. The solid lines are estimated functions and the dashed line represents the true function.}
 \label{fig:lark_weakness}
\end{figure}


\section{L\'{e}vy adaptive B-spline regression}\label{sec:newmodel}

We consider a general type of basis function as the generating elements of an overcomplete system instead of specific kernel functions such as  Haar, Laplacian, and Gaussian. The LABS model uses B-spline basis functions which can all systematically express jumps, sharp peaks,  and smooth parts of the function.

\subsection{B-spline Basis}
The B-spline basis function consists of piecewise $k$ degree polynomials with $k-1$ continuous derivatives. In general, the B-spline basis of degree $k$ can be derived utilizing the Cox-de Boor recursion formula:
\begin{equation}
\begin{aligned}
B^{*}_{0,i}(x)&:=\left\{{\begin{matrix}
1 & \mathrm {if} \quad t_{i}\leq x<t_{i+1}\\
0 & \mathrm {otherwise} \end{matrix}}\right.\\
B^{*}_{k,i}(x)&:={\frac {x-t_{i}}{t_{i+k}-t_{i}}}B^{*}_{k-1,i}(x)+{\frac {t_{i+k+1}-x}{t_{i+k+1}-t_{i+1}}}B^{*}_{k-1,i+1}(x),
\end{aligned}
\label{eq:bsp}
\end{equation}
where  $t_i$ are  called knots which must be in non-descending order $t_{i}\leq t_{i+1}$ \citep{de1972calculating}, \citep{cox1972numerical}. The B-spline basis of degree $k$,  $B^{*}_{k,i}(x)$ then needs $(k+2)$ knots, $(t_i, \ldots, t_{i+k+1})$. For convenience of notation, we redefine the B-spline basis of degree $k$ with  a knot sequence $\boldsymbol{\xi}_k := (\xi_{k,1}, \ldots, \xi_{k,k+2})$ as follows.
\begin{equation}
\begin{aligned}
B_{0}(x; \boldsymbol{\xi}_0)&:=\left\{{\begin{matrix}
1 & \mathrm {if} \quad \xi_{0,1} \leq x < \xi_{0,2}\\
0 & \mathrm {otherwise} \end{matrix}}\right.\\
B_{k}(x; \boldsymbol{\xi}_k)&:={\frac {x-\xi_{k,1}}{\xi_{k,(k+1)}-\xi_{k,1}}}B_{k-1}(x; \boldsymbol{\xi}^{\star}_k)+{\frac {\xi_{k,(k+2)}-x}{\xi_{k,(k+2)}-\xi_{k,2}}}B_{k-1}(x; \boldsymbol{\xi}^{\star \star}),
\end{aligned}
\label{eq:bsp2}
\end{equation}
where $\boldsymbol{\xi}^{\star}_k := (\xi_{k,1}, \xi_{k,2}, \ldots, \xi_{k, (k+1)})$ and  $\boldsymbol{\xi}^{\star \star}_k := (\xi_{k,2}, \xi_{k,3}, \ldots, \xi_{k, (k+2)})$. 

The B-spline basis functions can have a variety of shapes and smoothness determined by knot locations and the degree of it. For example, a B-spline basis function can be a piecewise constant (k = 0), linear $(k = 1)$, quadratic  $(k = 2)$ , and cubic $(k = 3)$ functions.  Furthermore, the B-spline basis with equally spaced knots has the symmetric form. These bases are called a Uniform B-splines. Examples of the B-spline basis functions of different degrees with equally spaced knots are shown in \autoref{fig:bsp}.

\begin{figure}[ht!]
	\centering
	\includegraphics[width=10cm]{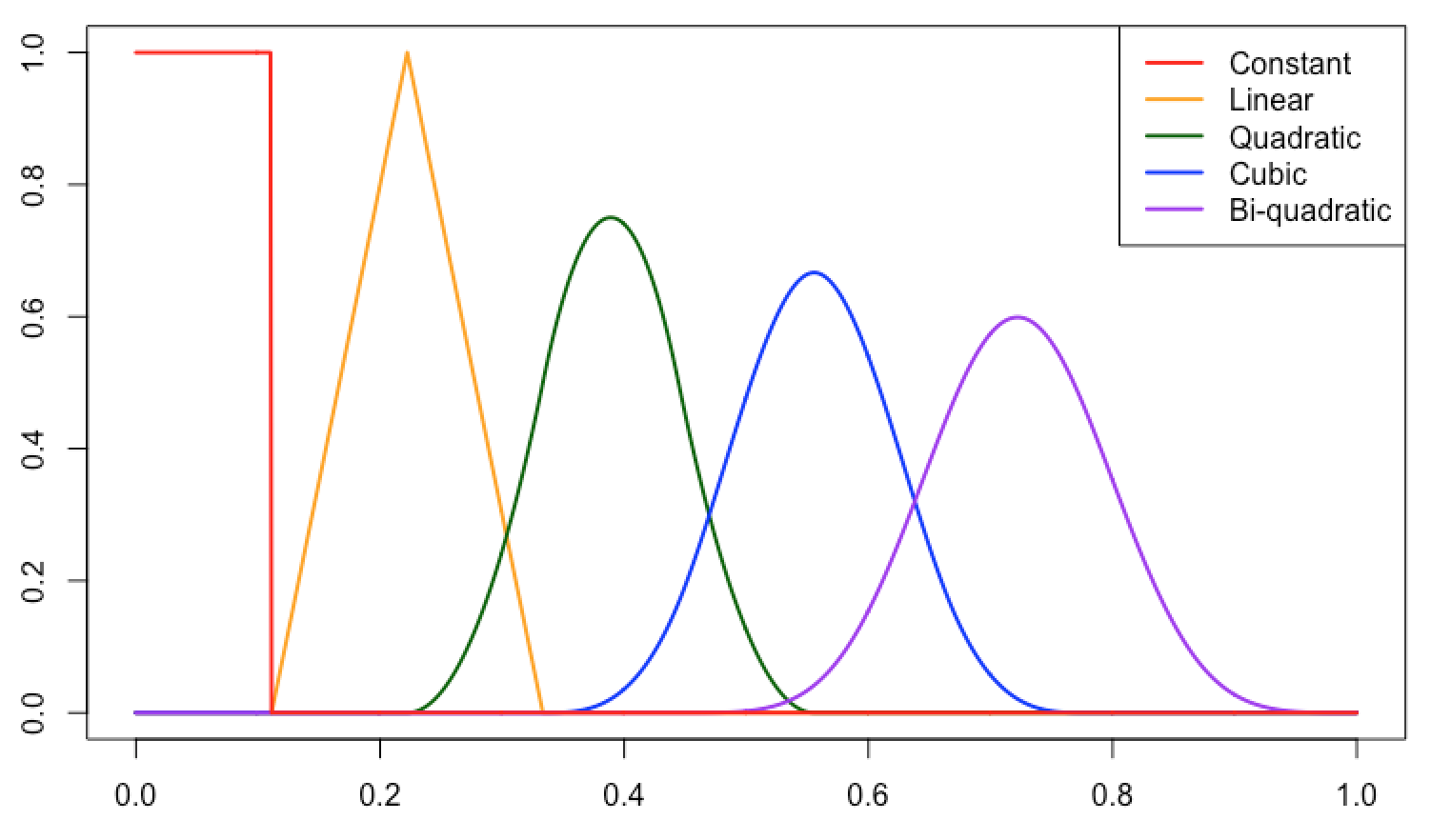}
	\caption{Different shapes of the B-spline basis function by increasing the degree $k$}
	\label{fig:bsp}
\end{figure}


\subsection{Model Specification}
The LARK model with one type of kernel can not estimate well functions with both continuous and discontinuous parts. To improve this,  we consider various a B-spline basis functions simultaneously for estimating all parts of the unknown function. The new model uses B-spline basis to systematically generate an overcomplete system with varying degrees of smoothness. For example, the  B-spline basis functions of degrees 0, 1 and 2 or more are for jumps, sharp peaks and smooth parts of the function, respectively.

We consider the mean function can be expressed as a random finite sum:
\begin{equation}
\eta(x) = \sum_{k \in S }\sum_{1 \leq l \leq J_k} B_{k}(x ; \boldsymbol{\xi}_{k,l}) \beta_{k,l}, \\
\label{eq:meanfunc}
\end{equation}
where $S$ denotes the subset of degree numbers of B-spline basis and $B_k(x; \boldsymbol{\xi}_k)$ is a B-spline basis of degree $k$ with knots, $\boldsymbol{\xi}_k  \in \mathcal{X}^{(k+2)} := \Omega$. Generating functions of the LARK model are replaced by the B-spline basis functions. $J_k$ has a Poisson distribution with $M_k > 0$ and $\{(\beta_{k,l}, \boldsymbol{\xi}_{k,l})\} \stackrel{iid}{\sim} \pi_k(d\beta_k, d \boldsymbol{\xi}_k) := \nu_k(d \beta_k, d \boldsymbol{\xi}_k)/\nu_k(\mathbb{R} \times \Omega$).  In this paper, we assume
\begin{equation*}
\pi_k(d \beta_k, d\boldsymbol{\xi}_k) = \mathcal{N}(\beta_k; 0,\phi_k^2 ) \, d\beta_k \cdot \mathcal{U}(\boldsymbol{\xi}_k;\mathcal{X}^{(k+2)})d \boldsymbol{\xi}_k.
\end{equation*}

The mean function can be also defined as 
\begin{equation}
\eta(x) \equiv \sum_{k \in S}  \int_{\Omega} B_k(x; \boldsymbol{\xi}_k) L_k(d\boldsymbol{\xi}_k).
\end{equation}
The stochastic integral representation of the mean function is determined by
\begin{equation*}
L_k \sim  \text{L\'{e}vy}(\nu_k(d \beta_k, d \boldsymbol{\xi}_k)), \quad \forall k \in S,
\end{equation*}
where $\nu_k(d \beta_k, d \boldsymbol{\xi}_k)$ is a  finite L\'{e}vy measure satisfying $M_k \equiv \nu_k(\mathbb{R} \times \Omega) < \infty$.
Although the L\'{e}vy measure $\nu_k$ satisfying \eqref{eq:cond1} may be infinite, the Poisson integrals and sums aboves are well defined for all bounded measurable compactly-supported $B_k(\cdot,\cdot)$ for which for all $k \in S$,
\begin{equation}
\int \int_{\mathbb{R} \times \Omega} (1 \wedge |\beta_k B_k(\cdot; \boldsymbol{\xi}_k)|) \nu_k(d \beta_k, d \boldsymbol{\xi}_k) < \infty.
\label{eq:PMcond2}
\end{equation}
 In this paper, we consider only finite L\'{e}vy measures in the proposed model. In other words, we restrict our attention to the L\'{e}vy measure of a compound Poisson process. The new model is  more complex than the LARK model with one kernel and expected to give a more accurate estimate of the regression function. It can estimate a mean function having both smooth and peak shapes. The proposed model can write in hierarchical form as
\begin{equation}
\begin{aligned}
Y_i|x_i &\stackrel{ind}{\sim} \mathcal{N}(\eta(x_i),\sigma^2), \quad i = 1,2, \cdots, n, \\
\eta(x) &= \beta_0 + \sum_{k \in S }\sum_{1 \leq l \leq J_k} B_{k}(x ; \boldsymbol{\xi}_{k,l}) \beta_{k,l}, \\
\sigma^2 &\sim \text{IG}\left(\dfrac{r}{2}, \dfrac{rR}{2}\right), \\
 J_k &\sim \text{Poi}(M_k),\\
 M_k &\sim \text{Ga}(a_{\gamma_k}, b_{\gamma_k)}, \\
 \beta_{k,l} &\stackrel{iid}{\sim} \mathcal{N}(0, \phi^2_{k}), \quad l = 1,2, \cdots, J_k, \\
 \boldsymbol{\xi}_{k,l}  &\stackrel{iid}{\sim} \mathcal{U}(\mathcal{X}^{(k+2)}), \quad l = 1,2, \cdots, J_k,
 \end{aligned}
 \label{eq:mdsp}
 \end{equation}
for $k \in S$. We set $\beta_0 = \overline{Y}$ and $\phi_k = 0.5 \times (\max_i\{Y_i\} - \min_i\{Y_i\})$.


\subsection{Support of LABS model}\label{subsec:basicideas}

In this section, we present three theorems on the support of the LABS model. We first define the modulus of smoothness and Besov spaces. 
\begin{definition}
Let $0 < p \leq \infty$ and $h > 0$. For $f \in L^p(\calX)$, the $r$th order modulus of smoothness of $f$ is defined by
 $$
 \omega_r(f,t)_p := \sup_{h \leq t} \|\Delta_h^r f\|_p,
 $$
where $\Delta_h^r f(x) = \sum_{k=0}^r \binom{r}{k}(-1)^{r-k} f(x + kh)$ for $x \in \calX$ and $x + kh \in \calX$.
\end{definition}
If $r = 1$, $\omega_1(f,t)_p$ is the modulus of continuity. There exist equivalent definitions in defining Besov spaces. We follow \cite{devore1993constructive}[2.10, page 54].

\begin{definition}
Let $\alpha > 0$ be given and let $r$ be the smallest integer such that $r > \alpha$. For $0 < p,q < \infty$, the Besov space  $\mathbb{B}^\alpha_{p,q}$ is the collection of all functions $f \in L_p(\calX)$ such that
$$
|f|_{\mathbb{B}^\alpha_{p,q}} = \left( \int_0^\infty [t^{-\alpha}\omega_r(f,t)_p]^q \frac{dt}{t} \right)^{1/q}
$$
is finite. The norm on  $\mathbb{B}^\alpha_{p,q}$ is  defined as 
$$
\|f\|_{\mathbb{B}^\alpha_{p,q}} = \|f\|_{p} + |f|_{\mathbb{B}^\alpha_{p,q}}.
$$
\end{definition}
The Besov space is a general function space depending on the smoothness of functions in $L_p(\calX)$ and especially can allow smoothness of spatially inhomogeneous functions, including spikes and jumps. The Besov space has three parameters, $\alpha$, $p$, and $q$, where $\alpha$ is the degree of smoothness, $p$ represents that $L_p(\Omega)$ is the function space where smoothness is measured, and $q$ is a parameter for a finer tuning on the degree of smoothness.

\begin{theorem}\label{thm:1}
For fixed $k \in S$ and $\boldsymbol{\xi}_k \in \calX^{(k+2)}$,  the B-spline basis $B_k(x; \boldsymbol{\xi}_k)$ falls in $\mathbb{B}_{p,q}^\alpha(\calX)$ for all $1 \leq p, q < \infty$ and $\alpha < k + 1/p$. 
\label{thm:bspinbsv}
\end{theorem}
The proof is given in \autoref{sec:appendA}. For instance, the B-spline basis with degree 0 satisfies $B_k(\cdot , \boldsymbol{\xi}_k) \in \mathbb{B}_{p,q}^\alpha$ for $\alpha < 1/p$, the B-spline basis with degree 1 is in $ \mathbb{B}_{p,q}^\alpha$ for $1 + 1/p$ and the B-spline basis with degree 2 falls in $ \mathbb{B}_{p,q}^s$ for $2 + 1/p$. 

The following theorem describes the mean function of the LABS model, $\eta$, is in a Besov space with smoothness corresponding to degrees of B-spline bases used in the LABS model. The proof of the theorem closely follows that of \cite{wolpert2011stochastic}. The  proof of \autoref{thm:2} is  given  in  \autoref{sec:appendA}. 
\begin{theorem} \label{thm:2} 
Suppose $\mathcal{X}$ is a compact subset of $\mathbb{R}$. Let $\nu_k$ be a L\'{e}vy measure  on $\mathbb{R} \times \mathcal{X}^{(k+2)}$ that satisfies the following integrability condition,
\begin{equation}
\int \int_{\mathbb{R} \times \mathcal{X}^{(k+2)}} (1 \wedge |\beta_k|) \nu_k (d\beta_k, d \boldsymbol{\xi}_k) < \infty.
\label{eq:PMcond1}
\end{equation}
and  $L_k \sim$ L\'{e}vy$(\nu_k)$ for all $k \in S$. Define the mean function of the LABS model, $\eta(\cdot) = \sum_{k \in S} \int_{\mathcal{X}^{(k+2)}} B_k(x; \boldsymbol{\xi}_k) L_k(d \boldsymbol{\xi}_k)$ on $\mathcal{X}$ where $B_k(x; \boldsymbol{\xi}_k)$ satisfies \eqref{eq:PMcond1} for each fixed $x \in \mathcal{X}$. Then, $\eta$ has the convergent series
\begin{equation}
\eta(x) = \sum_{k \in S} \sum_{l}^{} B_{k} (x; \boldsymbol{\xi}_{k,l}) \beta_{k,l}
\label{eq:PMdef}
\end{equation}
where $S$ is a finite set including degree numbers of B-spline basis. Furthermore, $\eta$ lies in the Besov space $\mathbb{B}_{p,q}^\alpha(\mathcal{X})$ with $\alpha < \min(S) + \frac{1}{p}$ almost surely.
\end{theorem}
 For example, if a zero element is included in $S$  then the mean function of the LABS, $\eta$ falls in $\mathbb{B}^\alpha_{p,q}$ with  $\alpha < \frac{1}{p}$ almost surely, which consists of functions no longer continuous. If $S = \{3,5,8\}$, then, $\eta$ belongs to $\mathbb{B}^\alpha_{p,q}$  with $\alpha < 3 + \frac{1}{p}$ almost surely. Moreover, it is highly possible that the function spaces for the LABS model may be larger than those of the LARK model using one type of kernel function. Specifically, the mean function for the LABS model with $S = \{0, 1\}$ falls in $\mathbb{B}^\alpha_{p,p}$ with $\alpha < \frac{1}{p}$ almost surely. If that of the LARK model using only one Laplacian kernel falls in $\mathbb{B}^\alpha_{p,p}$ with $\alpha < 1 + \frac{1}{p}$ , then the function spaces of the LABS model with given $\alpha < \frac{1}{p}$ are larger than those of the LARK model for the range of smoothness parameter, $\frac{1}{p} < \alpha < 1 + \frac{1}{p}$ by the properties of the Besov space.

The next theorem shows that the prior distribution of our LABS model has sufficiently large support on the Besov space $\mathbb{B}^\alpha_{p,q}$ with $1 \leq p, q < \infty$ and $\alpha > 0$. For $\eta_0 \in \mathbb{B}^\alpha_{p,q}(\mathcal{X})$, denote the ball around $\eta_0$ of radius $\delta$,
$$
\bar{b}_\delta(\eta_0) = \{\eta: \| \eta - \eta_0 \|_{p} < \delta \}
$$
where $\| \cdot \|_{p}$ is a $L_p$ norm. \ech The proof of \autoref{thm:3} is given in \autoref{sec:appendA}.
\begin{theorem}\label{thm:3}
Let $\calX$ be a bounded domain in $\mathbb{R}$. Let $\nu_k$ be a finite measure on $\mathbb{R} \times \mathcal{X}^{(k+2)}$  and  $L_k \sim Levy(\nu_k)$ for all $k \in S$. Suppose $\eta$ has a prior $\Pi$ for the LABS model \eqref{eq:mdsp}.  Then, $\Pi(\bar{b}_\delta(\eta_0)) > 0$ for every $\eta_0 \in  \mathbb{B}^\alpha_{p,q}(\mathcal{X})$ and all $\delta > 0$.
\end{theorem}


\section{Algorithm}\label{sec:algorithm}

Based on the prior specifications and the likelihood function,  the joint posterior distribution of the LABS model \eqref{eq:mdsp} is
	\begin{align}
	 \label{eq:mdsp1}
	 [\boldsymbol{\beta}, \boldsymbol{\xi},\boldsymbol{J},\boldsymbol{M},\sigma^2 \,|\, \boldsymbol{Y}] & \propto  [\boldsymbol{Y}\,|\,\eta ,\sigma^2] \times [\boldsymbol{\beta}, \boldsymbol{\xi} \,|\, \boldsymbol{J}] \times [\boldsymbol{J} \,|\, \boldsymbol{M}] \times [\boldsymbol{M}] \times[\sigma^2]  \nonumber\\
	& \propto \left[ (\sigma^2)^{-n/2} \exp \left\{-\dfrac{1}{2 \sigma^2}\sum_{i=1}^{n}(Y_i - \beta_0 - \sum_{k \in S}\sum_{l=1}^{J_k} B_{k}(x_i ; \boldsymbol{\xi}_{k,l}) \beta_{k,l})^2 \right \} \right]   \nonumber \\
	& \times \prod_{k \in S}\left[ \exp \left\{ -\dfrac{1}{2 \sigma^2_{k}} \sum_{l=1}^{J_k}\beta_{k,l}^2 \ \right\} \right] \times \prod_{k \in S} \left[ \dfrac{1}{|\mathcal{X}^{(k+2)}|^{J_k}} \prod_{l=1}^{J_k} I(\boldsymbol{\xi}_{k,l} \in \mathcal{X}^{(k+2)} )\right]  \nonumber \\
	& \times \prod_{k \in S}  \left[\dfrac{M_k^{J_k}}{J_k !} \exp\{-M_k\}  \right]\times \prod_{k \in S}\left[ M_k^{a_{\gamma_k} -1} \exp\{-b_{\gamma_k} M_k\}\right]  \nonumber \\
	&\times\left[ (\sigma^2)^{-\frac{r}{2} +1} \exp \left\{-\dfrac{rR}{2\sigma^2}\right\} \right].
	\end{align}
The parameters $\boldsymbol{\beta}$ and $\boldsymbol{\xi}$ of the LABS model have varying dimensions as $J_k$ is a random variable. We use the Reversible Jump Markov Chain Monte Carlo (RJMCMC) algorithm  \citep{green1995reversible}  for the posterior computation.

We consider three transitions in the generation of posterior samples: (a) the addition of basis functions and coefficients; (b) the deletion of basis functions and coefficients; (c) the relocation of knots which affects the shape of basis functions and coefficients. Note that in step (c) the  numbers of basis functions and coefficients do not change. We call such move types birth step, death step, and relocation step, respectively. A type of move is determined with probabilities $p_b$, $p_d$ and $p_w$ with $p_b + p_d + p_w = 1$, where $p_b$, $p_d$ and $p_w$ are probabilities of choosing the birth, death, and relocation steps, respectively.

Let us denote $\theta_{k,l} = (\beta_{k,l}, \boldsymbol{\xi}_{k,l})$ by an element  of $\boldsymbol{\theta}_k = \{\theta_{k,1}, \theta_{k,2}, \ldots, \theta_{k,j},\ldots, \theta_{k,J_k}\}$, where each  $ \boldsymbol{\xi}_{k,l}$  has the  $(k+2)$ dimensions. In general, the acceptance ratio of the RJMCMC can be expressed as
\begin{equation*}
A = \min\left[1,  \text{(likelihood ratio)} \times \text{(prior ratio)} \times \text{(proposal ratio)} \times \text{(Jacobian)} \right].
\end{equation*}
In our problem the acceptance ratio for each move types is given by
\begin{equation}
\label{eq:AR}
A = \min\left[1, \frac{L(\mathbf{Y} | \boldsymbol{\theta}_k',J'_k) \,\Pi(\boldsymbol{\theta}_k' | J'_k) \Pi(J'_k) q(\boldsymbol{\theta}_k | \boldsymbol{\theta}_k')}{L(\mathbf{Y} | \boldsymbol{\theta}_k, J_k) \,\Pi(\boldsymbol{\theta}_k | J_k) \Pi(J_k) q(\boldsymbol{\theta}_k' | \boldsymbol{\theta}_k)}\right],
\end{equation}
where $\boldsymbol{\theta}_k$ and $J_k$ refer to the current model parameters and the number of basis functions in the current state,  $\boldsymbol{\theta}_k'$ and  $J_k'$ denote the proposed model parameters  and the number of basis functions of the new state. Here, the Jacobian is 1 in all move types. $q(\boldsymbol{\theta}_k' | \boldsymbol{\theta}_k)$ is the jump proposal distribution that proposes a new state $\boldsymbol{\theta}_k'$ given a current state $\boldsymbol{\theta}_k$. Specifically, we choose the following jump proposal density proposed by \cite{lee2020bayesian} for each move steps:
\begin{align*}
q_b(\boldsymbol{\theta}_k' | \boldsymbol{\theta}_k) &= p_b \times b(\theta_{k,J_k + 1})\times \frac{1}{J_k+1}, \\
q_b(\boldsymbol{\theta}_k' | \boldsymbol{\theta}_k) &= p_d \times \frac{1}{J_k}, \\
q_w(\boldsymbol{\theta}_k' | \boldsymbol{\theta}_k) &= p_w \times q(\theta'_{k,r} | \theta_{k,r}),
\end{align*}
where $b(\cdot)$ is a candidate distribution which proposes a new element. For death and change steps, a  randomly chosen $r$th element of $\boldsymbol{\theta}_k$ is deleted and rearranged,  respectively. The details regarding updating schemes of each move steps are as follows.

\begin{enumerate}[label=(\alph*)]
	\item \textbf{[Birth step]} Assume that the current model is composed of $J_k$ basis functions.  If the birth step is selected, a new basis function $B_{k, J_{k+1}}$ and $\theta_{k,J_{k+1}}$ is accepted with the acceptance ratio
	$$
\min\left[1, \dfrac{L(\mathbf{Y} | \boldsymbol{\theta}_k',J_k')}{L(\mathbf{Y} | \boldsymbol{\theta}_k,J_k)} \times \dfrac{\pi(\theta_{k, J_k+1}) M_k}{(J_k+1)} \times \dfrac{p_d/(J_k+1)}{(p_b \times b(\theta_{k, J_k+1}))/(J_k+1)}\right].
	$$
Especially, a coefficient $\beta_{k, J_{k+1}}$ and an ordered knot set $\boldsymbol{\xi}_{k,J_{k+1}}$ are drawn from their generating distributions and added at the end of $(\beta_{k,1},\ldots, \beta_{k, J_k})$ and $(\boldsymbol{\xi}_{k,1}, \ldots,\boldsymbol{\xi}_{k,J_{k}} )$. When $J_k = 0$,  the birth step must be exceptionally selected until $J_k$ becomes one.

	\item \textbf{[Death step]} If the death step is selected, a  $r$th element, $\theta_{k, r}$  uniformly chosen is removed from the existing set of basis functions, coefficients and ordered knot sets. We can find out the acceptance ratio for a death step similarly. The acceptance ratio is given by
$$
\min \left[1, \dfrac{L(\mathbf{Y} | \boldsymbol{\theta}_k',J_k')}{L(\mathbf{Y} | \boldsymbol{\theta}_k,J_k)} \times \dfrac{J_k}{\pi(\theta_{k,r}) M_k} \times \dfrac{(b(\theta_{k,r}) \times p_b)/J_k}{p_d /J_k}\right].
$$

\item \textbf{[Relocation step]}  Unlike the other steps, the relocation step \ech keeps the numbers of basis functions  or coefficients  or ordered knot sets fixed. Therefore, the updating scheme of this step is  a Metropolis-Hastings within Gibbs sampler.  If the relocation step is selected,  a current location $\theta_{k, r}$  is moved to a new location $\theta'_{k, r}$ generated by proposal distributions  with the acceptance ratio \eqref{eq:acprob}.  Particularly,  since knots of basis function must be in non-descending order, ${\xi}_{k,r,i}$ which is the $i$th element of an ordered knot set is sequentially replaced with a new knot location ${\xi}'_{k,r,i}$ generated by $\mathcal{U}({\xi}_{k,r,i-1}, {\xi}_{k,r,i+1}), i = 1, \ldots, (k+2)$, where ${\xi}_{k,r,0}$ and ${\xi}_{k,r,k+1}$  are boundary points of $\mathcal{X}$. That is,
each element of a specific knot set  $\boldsymbol{\xi}_{k, r} = (\xi_{k,r,1},\ldots, {\xi}_{k,r,k+2})$ is  moved to new knot locations $\boldsymbol{\xi}'_{k, r} = (\xi'_{k,r,1},\ldots, {\xi}'_{k,r,k+2})$ in turn. The corresponding acceptance ratio is given by
\begin{equation}
    \min\left[1,\dfrac{L(\mathbf{Y} | \boldsymbol{\theta}_k',J_k')}{L(\mathbf{Y} | \boldsymbol{\theta}_k,J_k)} \times \dfrac{\pi(\theta'_{k,r})}{\pi(\theta_{k,r}) } \times \dfrac{q_w(\theta_{k,r} | \theta'_{k,r})}{q_w(\theta'_{k,r} | \theta_{k,r})} \right].
    \label{eq:acprob}
\end{equation}
When using an independent proposal distribution (i.e. $q_w(\theta'_{k,r} | \theta_{k,r}) = \pi(\theta'_{k,r})$), the acceptance ratio can reduce to
$$
\min\left[1,\dfrac{L(\mathbf{Y} | \boldsymbol{\theta}_k',J_k')}{L(\mathbf{Y} | \boldsymbol{\theta}_k,J_k)} \right].
$$
Finally, $\beta'_{k,r}$ is sampled from its conditional posterior distribution by using the Gibbs sampling.
\end{enumerate}
The posterior samples of $\sigma$ and  $M_k$  can be generated from their conditional posterior distributions. See \autoref{sec:appendC}.
The pseudo-code for the proposed strategy is given in Algorithm \ref{alg:labs}.

	\begin{algorithm}
	\caption{A reversible jump MCMC algorithm for LABS}
	\label{alg:labs}
	\begin{algorithmic}[1]
	\Procedure{LABS}{$S$} \Comment{$S$: set of degree numbers}

	    \State Initialize parameters $\boldsymbol{J}, \boldsymbol{\beta},\boldsymbol{\xi},\boldsymbol{M},\sigma^2$ from prior distributions.
	  \For {iteration $i=1$ to $N$}
	    \For {$k=1$ to $|S|$} \Comment{$\boldsymbol{J} :=\{\boldsymbol{J}_1, \ldots,\boldsymbol{J}_k, \ldots,\boldsymbol{J}_{|S|}\}$}
	        \State Update $(\boldsymbol{J}_k, \boldsymbol{\beta}_k, \boldsymbol{\xi}_k)$ through a reversible jump MCMC.
	        \State Sample $M_k$ from the full conditional  $\pi(M_k | \text{others})$. \Comment{Gibbs step}
		\EndFor
	  \State Sample $\sigma^2$ from the full conditional $\pi(\sigma^2 | \boldsymbol{\beta}, \boldsymbol{\xi},\boldsymbol{J},\boldsymbol{M}, \boldsymbol{y})$. \Comment{Gibbs step}
	  \State Store $i$th MCMC samples.
	\EndFor
	\EndProcedure
	\end{algorithmic}
	\end{algorithm}


\section{Simulation Studies}\label{sec:simulation}
In this section, we evaluate the performance of the LABS model \eqref{eq:mdsp}  and competing methods on simulated data sets. First, we apply the proposed method to four standard examples: Bumps, Blocks, Doppler and Heavisine test functions introduced by  \cite{donoho1994ideal}. Second, we consider three functions that we created ourselves with jumps and peaks to assess the practical performance of the proposed model.

The simulated data sets are generated from equally spaced $x$'s on $\mathcal{X} = [0,1]$ with  sample sizes $n = 128$ and $512$. Independent normally distributed noises $\mathcal{N}(0,\sigma^2)$ are added to the true function $\eta(\cdot)$.  The root signal-to-noise ratio (RSNR) is defined as 
$$ 
\text{RSNR} := \sqrt{\frac{\int_{\mathcal{X}}(f(x) - \bar{f})^2 \, dx}{\sigma^2}},
$$ 
where $\bar{f}:= \frac{1}{|\mathcal{X}|}\int_{\mathcal{X}} f(x) \, dx$ and set at 3, 5 and 10. We also run the LABS model for 200,000 iterations, with the first 100,000 iterations discarded as burn-in and retain every 10th sample.
For comparison between the methods, we compute the mean squared errors of all methods using 100 replicate data sets for each test function. The average of the posterior curves is used for the estimate of the test function.
\begin{equation*}
\text{MSE} = \frac{1}{n}\sum_{i=1}^{n}(\eta(x_i)-\hat{\eta}(x_i))^2.
\end{equation*}


\subsection{Simulation 1 : DJ test functions}\label{subsec:ex1}

We carry out a simulation study using the benchmark test functions suggested by \cite{donoho1994ideal} often used in the field of wavelet and nonparametric function estimation. The Donoho and Johnstone test functions consist of four functions called Bumps, Blocks, Doppler and Heavisine. These test functions are composed of various shapes such as sharp peaks (Bumps), jump discontinuities (Blocks), oscillating behavior (Doppler) and jumps/peaks in smooth functions (Heavisine) (See \autoref{fig:ex1}).

\begin{figure}[ht!]
	\centering
	\includegraphics[width=0.8\textwidth]{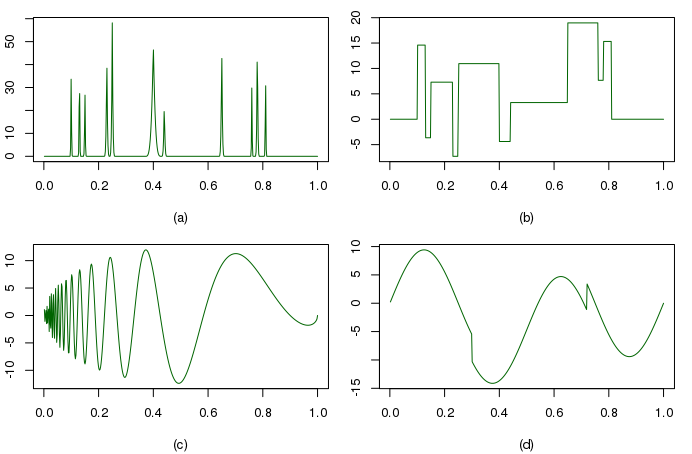}
	\caption{The Donoho and Johnstone test functions: (a) Bumps, (b) Blocks, (c) Doppler and (d) Heavisine}
	\label{fig:ex1}
\end{figure}

The hyperparameters and types of basis functions displayed in \autoref{tab:ex1-hp} were used in \eqref{eq:mdsp}.
For Bumps and Doppler, the parameter $r$ of prior distribution for $\sigma^2$ was set to 100 to speed up convergence. We also took account of the combinations of a B-spline basis based on the shapes of test functions.

\begin{table}[ht!]
\centering
\begin{tabular}{cccccc}
	\toprule
 & S & $r$ & $R$ & $a_{\gamma_{k}}$ & $b_{\gamma_{k}}$  \\
	\hline
	Bumps & \{1\} & 100 &  0.01& 1& 1\\
	Blocks & \{0\} & 0.01 & 0.01 & 1& 1\\
	Doppler & \{1,2\} & 100 & 0.01 & 1& 1\\
	Heavisine & \{0,2\} & 0.01 & 0.01 & 1&1\\
	 \bottomrule
\end{tabular}
\caption{The values of hyperparameters of proposed model for each test function}
\label{tab:ex1-hp}
\end{table}

We compared our model with a variety of methods such as B-spline curve of degree 2 with 50 knots (denoted as BSP-2), Local polynomial regression with automatic smoothing parameter selection (denoted by LOESS), Smoothing spline with smoothing parameter selected by cross-validation (denoted by SS), Nadaraya–Watson kernel regression using the Gaussian kernel with  bandwidth $h$ which minimizes CV error (denoted by NWK), Empirical Bayes approach for wavelet shrinkage using a Laplace prior with Daubechies “least asymmetric” (la8) wavelets except for the Blocks example, where it uses the Haar wavelet; \cite{johnstone2005empirical} (denoted by EBW),
Gaussian process regression with the Radial basis or Laplacian kernel (denoted by GP-R or GP-L), Bayesian curve fitting using piecewise polynomials with $l = \#1, l_0 = \#2$; \cite{denison1998automatic} (denoted by BPP-$\#1$-$\#2$), Bayesian adaptive spline surfaces with degree $\#$; \cite{francom2018sensitivity} (denoted by BASS-$\#$), and L\'{e}vy adaptive regression with multiple kernels; \cite{lee2020bayesian} (denoted by LARMuK). These competitive models are implemented in R \citep{r2020r} with various packages: LOESS \citep{wang2016rpack}, Empirical Bayes thresholding \citep{silverman2005rpack}, Gaussian process \citep{karatzoglou2004rpack}, Bayesian curve fitting using piecewise polynomials \citep{feng2013rpack}, and Bayesian adaptive spline surfaces \citep{francom2016rpack}.

\begin{figure}[ht!]
    \begin{subfigure}{\textwidth}
    \centering
    \includegraphics[width=0.95\linewidth]{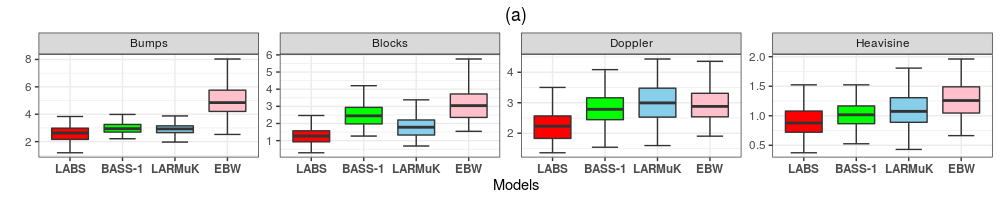}
    \label{fig:ex1-rsnr3}
    \end{subfigure}
    \begin{subfigure}{\textwidth}
    \centering
    \includegraphics[width=0.95\linewidth]{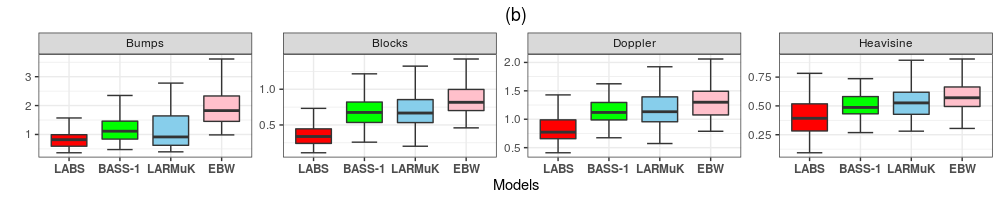}
    \label{fig:ex1-rsnr5}
    \end{subfigure}
    \begin{subfigure}{\textwidth}
    \centering
    \includegraphics[width=0.95\linewidth]{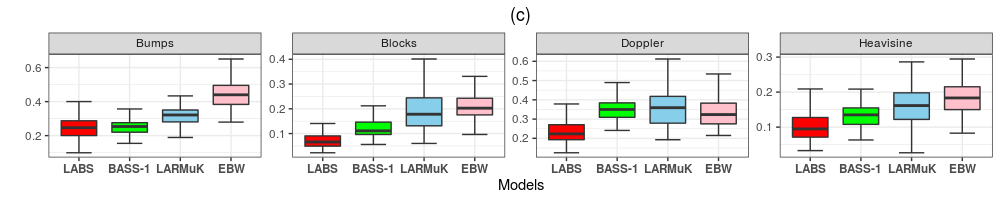}
    \label{fig:ex1-rsnr10}
    \end{subfigure}
\caption{Boxplots of MSEs from the simulation study with $n = 128$ and RSNR = (a) 3, (b) 5  and (c) 10}
\label{fig:ex1_bp128}
\end{figure}

Both \autoref{fig:ex1_bp128} and  \autoref{fig:ex1_bp512} show that the performance of our model is generally more accurate than other methods. The models in the two figures are selected by better outcomes from simulations. More detailed simulation results can be seen in \autoref{sec:appendB}.  \autoref{fig:ex1_bp128} shows that the LABS model is superior to others regardless of noise levels with $n = 128$. It also has the smallest average mean square errors for all test functions  except the Heavisine example with RSNR = 3.Similarly, for sample size $n = 512$,  the LABS model comes up with the best performance in  \autoref{fig:ex1_bp512} except for the Doppler function, where it is competitive. Our model removes high frequencies in the interval $[0, 0.1]$ and produces a smooth curve within the corresponding domain. On the contrary, due to a small number of data points in the Doppler example with $n = 128$, most models yield similar smooth curves in $[0, 0.1]$. As a result, the LABS model has an excellent numerical performance.
For Blocks example, LABS, in particular, yields the lowest average and standard deviation of mean square errors  in all scenarios. This suggests that our model has an excellent ability to find jump points. Furthermore, LABS has consistently better performance than B-spline regression using only one basis function for four simulated data sets since its overcomplete systems can be constructed by various combinations of B-spline basis functions. See \autoref{sec:appendB}.

\begin{figure}[ht!]
    \begin{subfigure}{\textwidth}
    \centering
    \includegraphics[width=0.95\linewidth]{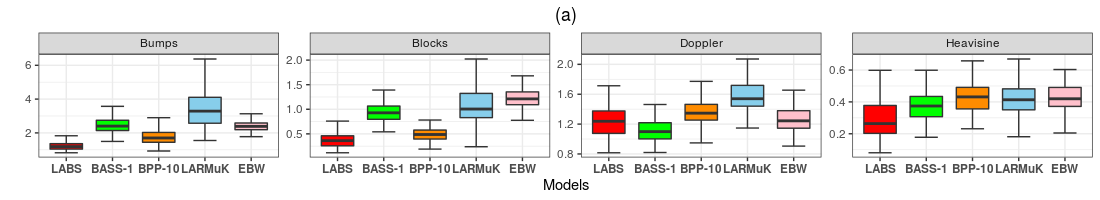}
    \label{fig:ex1-rsnr35}
    \end{subfigure}
    \begin{subfigure}{\textwidth}
    \centering
    \includegraphics[width=0.95\linewidth]{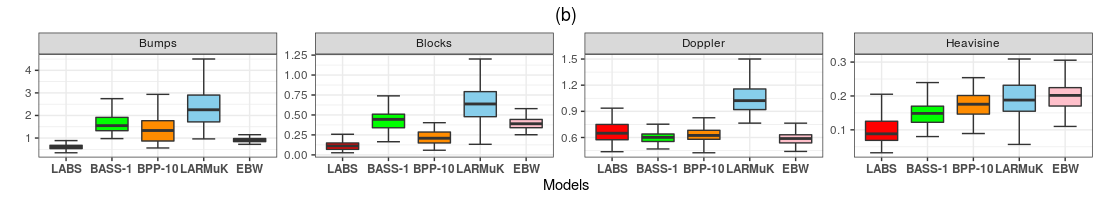}
    \label{fig:ex1-rsnr55}
    \end{subfigure}
    \begin{subfigure}{\textwidth}
    \centering
    \includegraphics[width=0.95\linewidth]{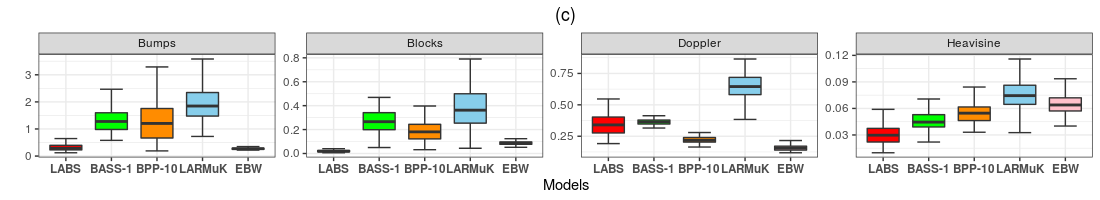}
    \label{fig:ex1-rsnr105}
    \end{subfigure}
\caption{Boxplots of MSEs from the simulation study with $n = 512$ and RSNR = (a) 3, (b) 5 and (c) 10}
\label{fig:ex1_bp512}
\end{figure}


\subsection{Simulation 2 : Smooth functions with jumps and peaks}
Our main interest lies in estimating smooth functions with either discontinuity such as jumps or sharp peaks or both. We design three test functions to assess the practical performance of the proposed method for our concerns. The first and second example is modified by adding some smooth parts, unlike the original version of the Bumps and Blocks of DJ test functions. Each test functions provided is given by
\begin{align*}
\eta_1(x) = &\dfrac{0.6}{0.92}[4 \text{ssgn}(x - 0.1) - 5\text{ssgn}(x - 0.13) + 5\text{ssgn}(x - 0.25) - 4.2\text{ssgn}(x - 0.4)\\
& + 2.1 \text{ssgn}(x - 0.44) + 4.3 \text{ssgn}(x - 0.65) - 4.2\text{ssgn}(x - 0.81) + 2] + 0.2 + \sin(8 \pi x), \\
\eta_{2}(x) = & [7K_{0.005}(x - 0.1) + 5K_{0.07}(x - 0.25)  + 4.2K_{0.03}(x - 0.4) + 4.3K_{0.01}(x - 0.65) \\
& + 5.1K_{0.008}(x - 0.78) + 3.1K_{0.1}(x - 0.9) ]+ \cos{(4\pi x)},
\end{align*}
where $\text{sgn}(x) = \text{I}_{(0,\infty)}(x) - \text{I}_{(-\infty,0)}(x)$, $\text{ssgn}(x) = 1 + \text{sgn}(x) /2$ and $K_w(x) := ( 1 + | x/w| )^{-4}$. Finally, we create a sum of jumps, peaks and some smoothness.   A formula for a final test function is
\begin{align*}
	\eta_{3}(x) &= 6\sin(4\pi x) +  7(1+\text{sgn}(x - 0.1)/2) - 7(1+\text{sgn}(x - 0.18)/2)  \\
      & - 2\text{sgn}(x-0.37) + 17K_{0.01}(x-0.5) -3\text{sgn}(x-0.72) + 10K_{0.05}(x-0.89).
\end{align*}
They are displayed in \autoref{fig:ex2}. We call in turn them modified Blocks, modified Bumps, and modified Heavisine, respectively.

\begin{figure}[ht!]
	\centering
	\includegraphics[width=\textwidth]{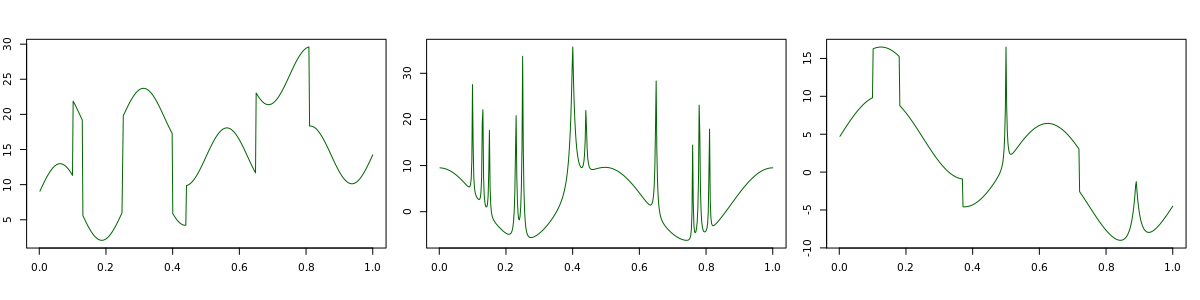}
	\caption{The three test functions used in the second simulation: Modified Blocks (left), Modified Bumps (center) and Modified Heavisine (right)}
	\label{fig:ex2}
\end{figure}

In these experiments, we use two or more types of B-spline basis as elements of overcomplete systems since three functions have different shapes, unlike previous simulation studies.
Hyperparameters are similar to the previous ones. All hyperparameters for the prior distributions are summarized in \autoref{tab:ex2-hp}. This time again, we only compare our model with BPP, BASS, EBW, and LARMuK models which have relatively good performance in some test functions of Simulation 1.
\begin{table}[ht!]
\centering
\begin{tabular}{cccccc}
	\toprule
 & S & $r$ & $R$ & $a_{\gamma_{k}}$ & $b_{\gamma_{k}}$  \\
	\hline
	Modified Blocks & \{0,2,3\} & 0.01 &  0.01& 1& 1\\
	Modified Bumps & \{1,2\} & 50 & 0.01 & 1& 1\\
	Modified Heavisine & \{0,1,2,3\} & 0.01 & 0.01 & 5&  1 \\
	 \bottomrule
\end{tabular}
\caption{Details of hyperparameters of the LABS used in second experiment}
\label{tab:ex2-hp}
\end{table}

\autoref{tab:ex2-128} furnishes that the LABS model has the best outcomes when the sample size is 128, which is difficult to estimate. Furthermore, when $n = 512$, we find out from both \autoref{tab:ex2-512} and  \autoref{fig:ex2_bp512} that the LABS model performs well in most cases with either the lowest or the second lowest average MSE values across 100 replicates. In particular, the LABS outperforms competitors in modified Blocks, irrespective of the sample size and noise levels as expected. Among all models, the worst performing method is the BASS-2 since it cannot estimate well many jumps or peak points for given test functions. \autoref{fig:ex2-comp} supports that the LABS model has the abilities to overcome the noise and adapt to  smooth functions with either discontinuity such as jumps or sharp peaks or both. 

\begin{table}[ht!]
 \centering
\adjustbox{max width=\textwidth}{%
 \footnotesize
	\begin{tabular}{cccccccccccc}
		\toprule
		\multirow{2}{*}{Model}&  \multicolumn{3}{c}{Modified Blocks}& \multicolumn{3}{c}{Modified Bumps}& \multicolumn{3}{c}{Modified Heavisine}\\
		\cmidrule{2-10}

		& RSNR=3 & RSNR=5 & RSNR=10 & RSNR=3 &RSNR=5 & RSNR=10 & RSNR=3 &RSNR=5 & RSNR=10 \\
		\midrule

EBW & 3.781(0.7407) & 1.538(0.3854) & 0.401(0.0684) & 3.921(1.0117) & 1.557(0.3191) & 0.445(0.1043) & 2.548(0.4738) & 1.326(0.2793) & 0.402(0.095) \\
  BPP-10 & 2.238(0.6436) & 0.951(0.2439) & 0.367(0.098) & 2.949(0.749) & 1.351(0.3244) & 0.631(0.2488) & 2.06(0.645) & 0.824(0.2125) & 0.287(0.0907) \\
  BPP-21 & 2.589(0.4787) & 1.336(0.2305) & 0.985(0.1714) & 3.777(0.9094) & 2.586(0.6268) & 2.39(0.6003) & 2.168(0.3821) & 1.228(0.2969) & 0.825(0.3458) \\
  BASS-1 & 2.283(0.5013) & 0.76(0.2194) & 0.172(0.0424) & 2.199(0.5625) & 0.858(0.1622) & 0.276(0.0483) & 2.013(0.5502) & 0.737(0.198) & 0.178(0.0418) \\
  BASS-2 & 4.038(0.6519) & 2.232(0.371) & 1.368(0.168) & 9.881(1.1796) & 7.944(0.7727) & 6.999(0.5772) & 3.275(0.3734) & 2.276(0.3877) & 1.378(0.2871) \\
  LARMuK & 2.158(0.5735) & 0.97(0.2352) & 0.298(0.0929) & 2.029(0.7688) & 0.822(0.2446) & 0.271(0.0944) & 1.721(0.4757) & 0.713(0.1912) & 0.219(0.075) \\
  LABS & \textbf{1.868(0.5982)} & \textbf{0.691(0.2022)} & \textbf{0.162(0.044)} & \textbf{2.01(0.6006)} & \textbf{0.803(0.1491)} & \textbf{0.248(0.0457)} & \textbf{1.589(0.5081)} & \textbf{0.635(0.1727)} & \textbf{0.172(0.0472)} \\
   \bottomrule
	\end{tabular}}
  \caption{Average of MSEs over 100 replications for three functions of Simulation 2 with $n = 128$. Estimated standard errors of MSEs are shown in parentheses}
  \label{tab:ex2-128}
\end{table}

\begin{table}[ht!]
 \centering
\adjustbox{max width=\textwidth}{
 \footnotesize
	\begin{tabular}{cccccccccccc}
		\toprule
		\multirow{2}{*}{Model}&  \multicolumn{3}{c}{Modified Blocks}& \multicolumn{3}{c}{Modified Bumps}& \multicolumn{3}{c}{Modified Heavisine}\\
		\cmidrule{2-10}

		& RSNR=3 & RSNR=5 & RSNR=10 & RSNR=3 &RSNR=5 & RSNR=10 & RSNR=3 &RSNR=5 & RSNR=10 \\
		\midrule

EBW & 1.525(0.228) & 0.644(0.1025) & 0.171(0.0245) & 1.487(0.2056) & 0.595(0.0847) & \textbf{0.171(0.0217)} & 1.106(0.1523) & 0.538(0.0775) & 0.153(0.0218) \\
  BPP-10 & 0.608(0.1362) & 0.256(0.055) & 0.133(0.0412) & 1.2(0.1951) & 0.623(0.2658) & 0.37(0.2954) & 0.602(0.123) & 0.237(0.0495) & 0.079(0.0183) \\
  BPP-21 & 0.869(0.1552) & 0.424(0.0683) & 0.283(0.0479) & 1.209(0.3056) & 0.937(0.3201) & 0.789(0.3442) & 0.679(0.1367) & 0.275(0.0562) & 0.134(0.0338) \\
  BASS-1 & 0.701(0.1489) & 0.296(0.0667) & 0.124(0.0455) & 0.927(0.147) & \textbf{0.442(0.0726)} & 0.197(0.063) & \textbf{0.561(0.126)} & 0.246(0.0423) & 0.083(0.0166) \\
  BASS-2 & 4.038(0.6519) & 2.232(0.371) & 1.368(0.168) & 9.881(1.1796) & 7.944(0.7727) & 6.999(0.5772) & 3.275(0.3734) & 2.276(0.3877) & 1.378(0.2871) \\
  LARMuK & 0.775(0.257) & 0.416(0.1317) & 0.213(0.0841) & 1.222(0.4575) & 0.805(0.3056) & 0.514(0.177) & 0.669(0.1749) & 0.323(0.0975) & 0.14(0.0456) \\
  LABS & \textbf{0.583(0.1718)} & \textbf{0.234(0.0696)} & \textbf{0.071(0.0325)} & \textbf{0.919(0.1397)} &0.46(0.0867) & 0.194(0.0443)  & 0.576(0.146) & \textbf{0.236(0.0575)} & \textbf{0.078(0.024)} \\
  \bottomrule
	\end{tabular}}
  \caption{Average of MSEs over 100 replications for three functions of Simulation 2 with $n = 512$. Estimated standard errors of MSEs are shown in parentheses}
  \label{tab:ex2-512}
\end{table}

\begin{figure}[ht!]
    \begin{subfigure}{\textwidth}
    \centering
    \includegraphics[width=0.95\linewidth]{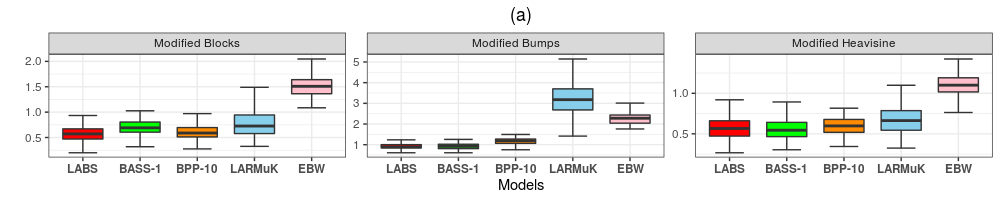}
    \label{fig:ex2-rsnr35}
    \end{subfigure}
    \begin{subfigure}{\textwidth}
    \centering
    \includegraphics[width=0.95\linewidth]{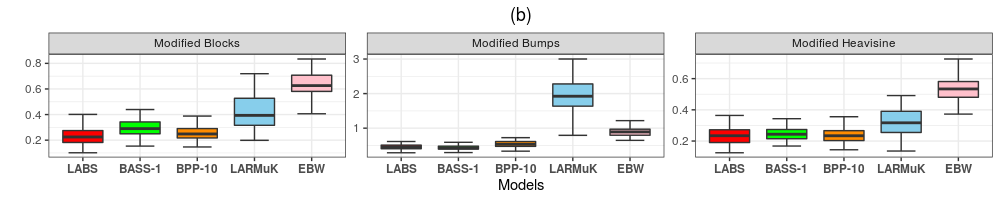}
    \label{fig:ex2-rsnr55}
    \end{subfigure}
    \begin{subfigure}{\textwidth}
    \centering
    \includegraphics[width=0.95\linewidth]{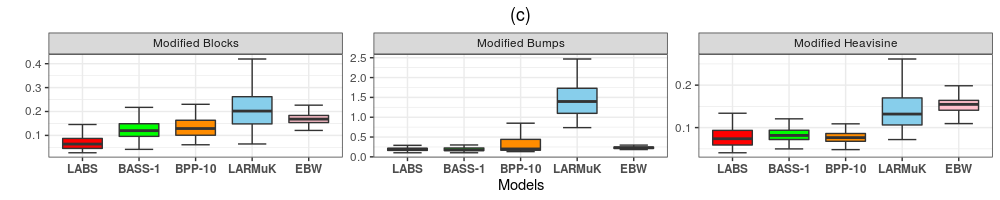}
    \label{fig:ex2-rsnr105}
    \end{subfigure}
\caption{Boxplots of MSEs from the second simulation with $n = 512$ and RSNR = (a) 3, (b) 5 and (c) 10}
\label{fig:ex2_bp512}
\end{figure}

\begin{figure}[ht!]
 \centering
 \includegraphics[width=\textwidth]{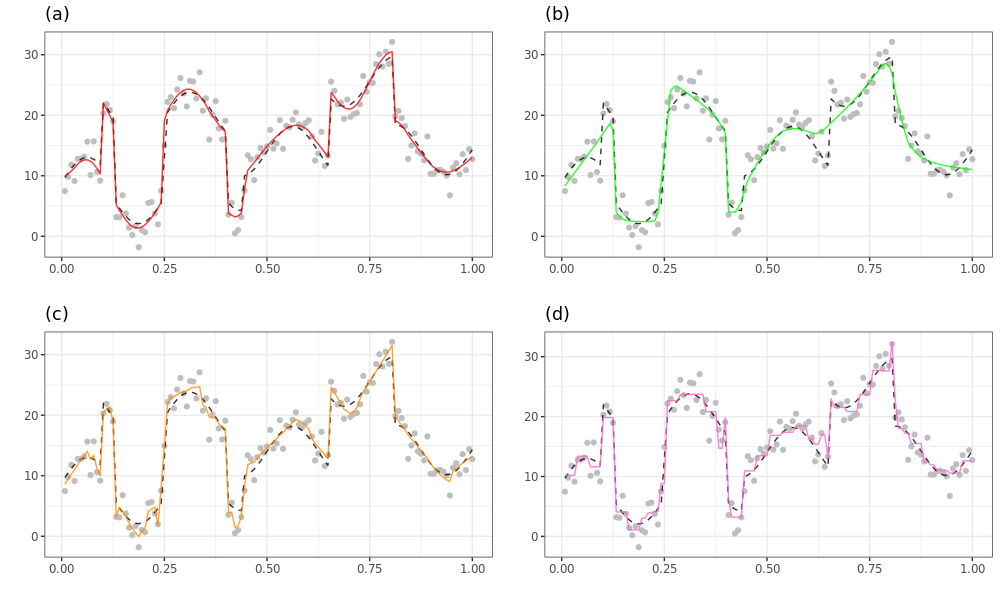}
 \caption{Comparisons of the estimates of a data set generated from the modified Blocks with $n = 128$ and RSNR = 3 using (a) LABS, (b) BASS-1, (c) BPP-10, and (d) EBW. Dashed lines represent true curves, solid lines represent estimates of curve.}
 \label{fig:ex2-comp}
 \end{figure}



\section{Real data applications}\label{sec:application}

We now apply LABS model \eqref{eq:mdsp} real-world datasets, including the minimum legal drinking age (MLDA) on mortality rate, the closing bitcoin price index, and the daily maximum value of concentrations of fine particulate matter (PM2.5) in Seoul. All the real data examples exhibit wildly varying patterns that may have jumps or peaks. These fluctuating patterns are expected to further illustrate the features of the LABS model.

In the following applications we set the hyperparameter values of the proposed model, LABS: $a_J = 5$ , $b_J = 1$, $r = 0.01$, and $R = 0.01$. In this analysis, we practically choose $S = \{0,1,2\}$ because the true curve of real data is unknown and it may have varying smoothness. 
We run it 200,000 times with a burn-in of 100,000 and thin by 10 to achieve convergence of the MCMC algorithm. Performance comparisons of our model and some rather good methods in the simulated studies are also conducted.


\subsection{Example 1: Minimum legal drinking age}\label{subsec:mlda}

The Minimum legal drinking age (MLDA) heavily affects youth alcohol consumption which has been a sensitive issue worldwide for policymakers.
In the past three decades, there have been many studies on the effect of legal access age to alcohol on death rates. The MLDA dataset collected from \cite{angrist2014mastering} contains death rates, a measure of the total number of deaths per 100,000 young Americans per year.

This data has been widely used to estimate the causal effect of policies on the minimum legal drinking age in the area of Regression Discontinuity Design (RDD). \autoref{fig:apps_rd} (a) highlights that the MLDA data might represent a piecewise smooth function with a single jump discontinuity at minimum drinking age of 21 referred to as cutoff in the RDD.  Specifically, each observation (or point) in \autoref{fig:apps_rd} corresponds to the death rate from all causes in the monthly interval and the number of all observations is 48.

 \begin{figure}[ht!]
 	\centering
 	\includegraphics[width=\textwidth]{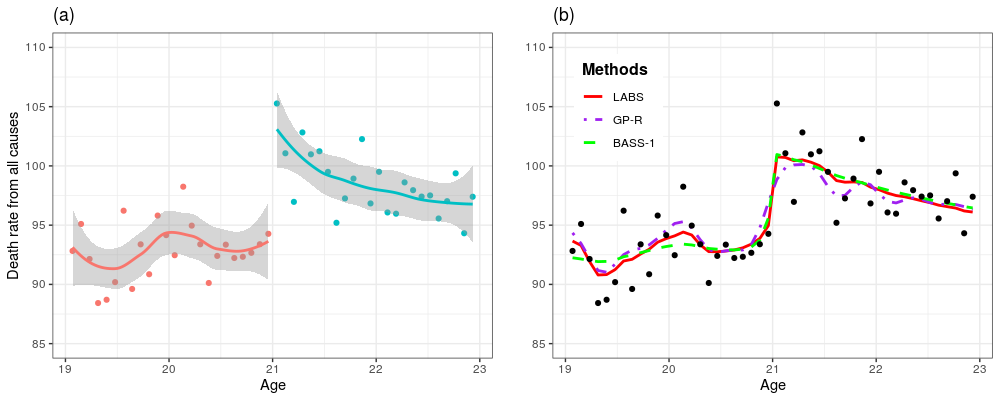}
 	\caption{(a) A piecewise curve fitting and (b) comparisons of the fitted posterior mean using BASS-1, GP-R and LABS for Minimum legal drinking age (MLDA) dataset}
 	\label{fig:apps_rd}
 \end{figure}

Using the MLDA data, we estimate unknown functions of death rates via LABS and several competing models including BPP-21, BASS-1, LARMuK, and GP-R.  \autoref{fig:apps_rd} (b) shows that three posterior mean estimates for an unknown function.  The solid curve denotes LABS, the dotted-dashed curve indicates GP-R, and the dashed curve represents BASS-1, respectively.

In \autoref{fig:apps_rd} (b), both LABS and BASS-1 provide similar posterior curves to the piecewise polynomial regression of \autoref{fig:apps_rd} (a). The estimated curves of them also have a jump point at 21.  While the estimated function of GP-R is smooth, the mean function for the LABS model has both smooth and jump parts.  We calculated the mean squared error with 10-folds cross-validation for comparison between methods.  The mean and standard deviation values of cross-validation prediction errors are given in \autoref{tab:apps_rd}. The smaller CV error rate of LABS implies that LABS has a better performance of estimating a smooth function with discontinuous points than the others.

\begin{table}[ht!]
\centering
\begin{tabular}{cccccc}
  \toprule
 & LABS &BASS-1 & BPP-21 & LARMuK & GP-R \\
  \hline
 Mean & \textbf{6.5851} & 6.7884 & 8.6643 & 8.35014 & 7.25693 \\
 Standard Deviation & 5.1838 & 4.98241 & 6.66641 & 6.45711 & 5.23563 \\
   \bottomrule
\end{tabular}
\caption{Mean and standard deviation for the error rate of 10-folds cross-validation on MLDA dataset.}
\label{tab:apps_rd}
\end{table}


\subsection{Example 2: Bitcoin prices on Bitstamp}\label{subsec:bitcoin}

Bitcoin is the best-known cryptocurrency based on Blockchain technology. The demands for bitcoin have increased globally because of offering a higher yield and easy access. The primary characteristic of bitcoin is to enable financial transactions from user to user on the peer-to-peer network configuration without a central bank. Unlimited trading methods and smaller market size than the stock market lead to  high volatility in the bitcoin price. We collected a daily bitcoin exchange rate (BTC vs. USD) on Bitstamp from January 1, 2017, to December 31, 2018. Bitcoin data (sourced from \url{http://www.bitcoincharts.com}) has 730 observations and 8 variables: date, open price (in USD), high price, low price, closing price, volume in bitcoin, volume in currency, and weighted bitcoin price.

\begin{figure}[ht!]
 \centering
 \includegraphics[width=\textwidth]{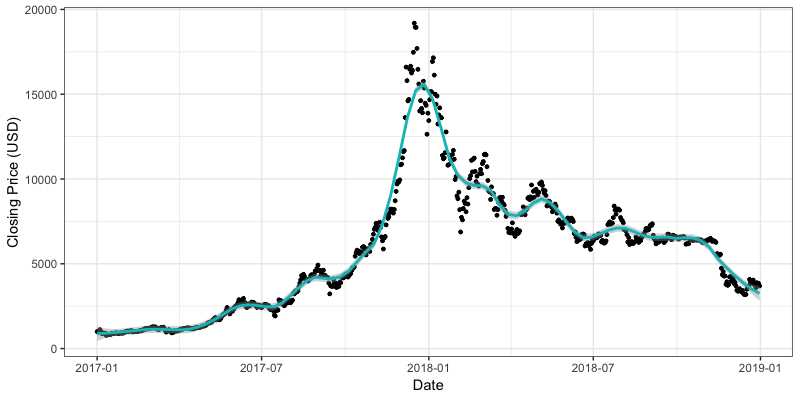}
 \caption{Daily bitcoin closing price with a smoothing line}
 \label{fig:apps_btc1}
\end{figure}

We also added LOESS (locally estimated scatterplot smoothing) regression line to a scatter plot of a daily closing price in \autoref{fig:apps_btc1}. The dataset shows locally strong upward and downward movements.
We apply LABS and other models to estimate the curve of daily bitcoin closing price. \autoref{fig:apps_btc2} illustrates the predicted curves of the LABS and competing models for approximating an unknown function of daily bitcoin closing price. There are no significant differences between the estimated posterior curves.

\begin{figure}[ht!]
 \centering
 \includegraphics[width=\textwidth]{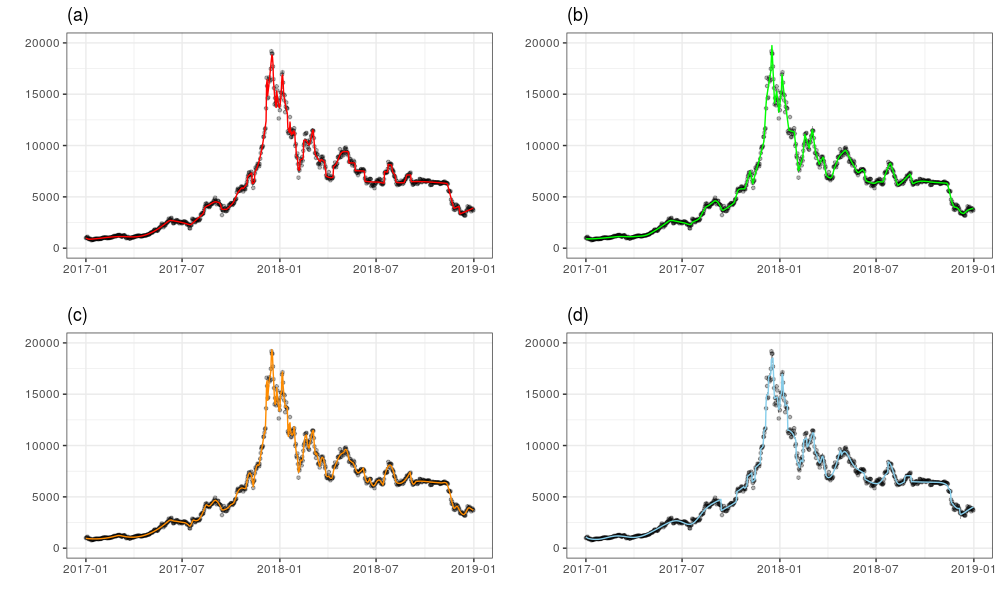}
 \caption{Posterior mean of $\eta$ on Bitcoin dataset using four models: (a) LABS, (b) BASS-1, (c) BPP-21 and (d) LARMuK}
 \label{fig:apps_btc2}
\end{figure}

Alternatively, we calculate cross-validation errors to assess model performances. The values of cross-validation errors are given in \autoref{tab:apps_btc}.  Both \autoref{tab:apps_btc} and \autoref{fig:apps_btc3} demonstrate that the LABS model provides more accurate function estimation and consistent performance through  both minimum mean and relatively low standard deviation values of the cross-validation errors.  They also indicate that the
Gaussian process is not proper in the cases with locally varying smoothness. We found that the LABS gives more reliable estimated functions that may have both discontinuous and smooth parts than other methods.

\begin{table}[ht!]
\centering
\begin{tabular}{cccccc}
  \toprule
 & LABS & BASS-1 & BPP-21 & LARMuK  & GP-R\\
  \hline
  Mean & \textbf{98014} & 109222 & 99937 & 149272 & 583046 \\
  Standard Deviation & 30057 & 30052.5& 29210 & 51128.7 & 137859.7 \\

   \bottomrule
\end{tabular}
\caption{Mean and standard deviation for the error rate of 10-folds cross-validation on Bitcoin dataset}
\label{tab:apps_btc}
\end{table}

\begin{figure}[ht!]
 \centering
 \includegraphics[width=0.5\linewidth]{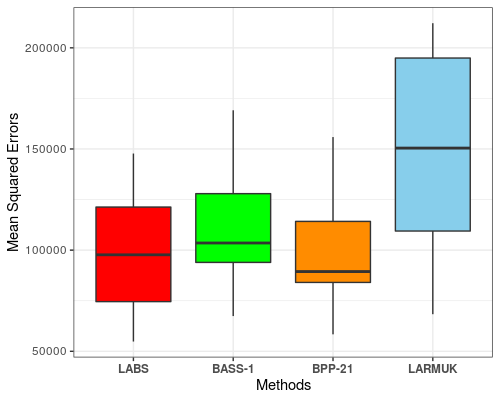}
 \caption{Boxplot of the cross-validation test error rate for the Bitcoin data.}
 \label{fig:apps_btc3}
\end{figure}


\subsection{Example 3: Fine particulate matter in Seoul}\label{subsec:pm}

The fine dust has become a national issue and its forecast received great attention from the media. A lot of research on fine particulate matter (PM2.5) have been carried out as it gained social attention.  According to the studies, Korea’s fine dust particles originated from within the country and external sources from China. Many factors cause PM2.5 concentration to rapidly rise or fall and make it difficult to accurately predict the behavior of it.

We estimate the unknown function of daily maximum concentrations of PM2.5 in Seoul. The PM2.5 dataset collected from the AIRKOREA (\url{https://www.airkorea.or.kr}) includes 1261 daily maximum values of PM2.5 concentration from January 1, 2015, to June 30, 2018. We removed all observations that have missing values.

\begin{figure}[ht!]
 \centering
 \includegraphics[width=\textwidth]{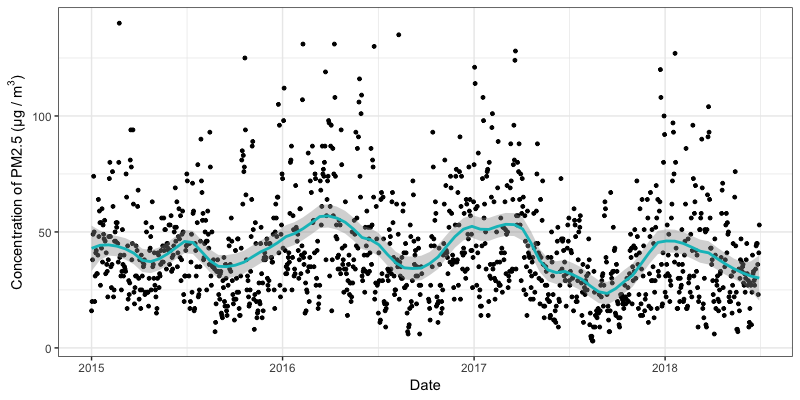}
 \caption{Daily maximum concentrations of PM2.5 in Seoul with a smoothing line}
 \label{fig:apps_pm25_1}
\end{figure}

\begin{figure}[ht!]
 \centering
 \includegraphics[width=\textwidth]{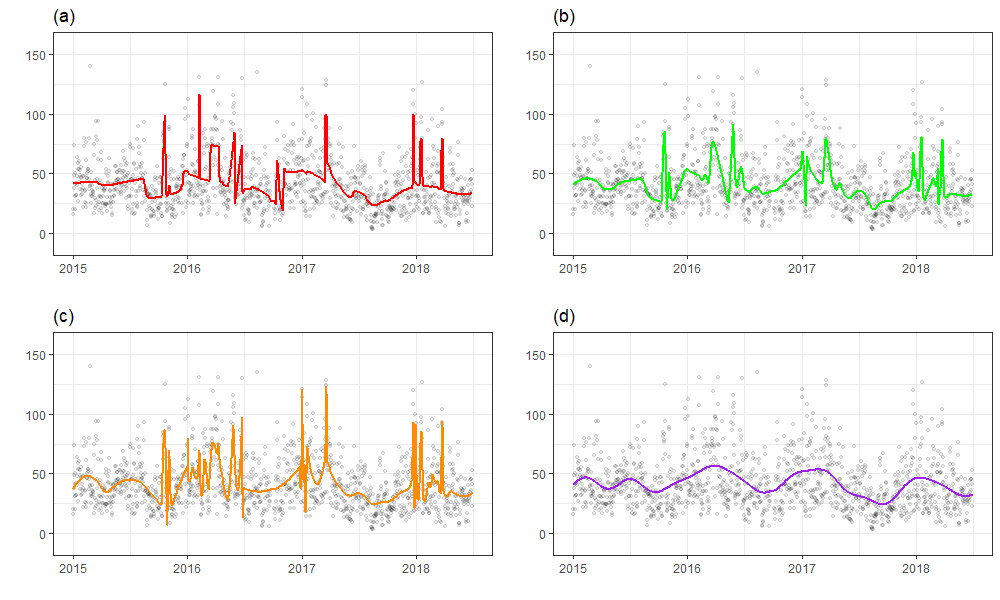}
 \caption{Posterior mean of the mean function on PM2.5 dataset using four models: (a) LABS, (b) BASS-1, (c) BPP-10 and (d) GP-R}
 \label{fig:apps_pm25_2}
\end{figure}

\autoref{fig:apps_pm25_1} displays daily fluctuations and seasonality. PM2.5 concentrations are higher in winter and spring than in summer and fall.
A LOESS smoothed line added in the figure does not reflect these features well. We take advantage of combinations of basis functions, $S = \{0, 1, 2\}$ to grasp such characteristics of PM2.5 data with multiple jumps and peak points. As shown in \autoref{fig:apps_pm25_2}, all four methods represent different estimated lines of the unknown mean function and pick features of the data up in their way. Interestingly, LABS, BASS-1, and BPP-10  react in different ways while they detect peaks, jumps, and smooth parts of PM2.5 data.

We also compute the average and standard deviation of the cross-validated errors of  LABS, BPP-10, BASS-1, LARMuK, and GP-R, which are  given in \autoref{tab:apps_pm}. The LABS model has the lowest cross-validation error among all methods. Moreover a comparably low standard deviation of LABS supports that it has a more stable performance for estimating any shape of functions due to using all three types of B-spline basis.

\begin{table}[ht!]
\centering
\begin{tabular}{cccccc}
  \toprule
 & LABS &BASS-1 & BPP-10 & LARMuK  & GP-R \\
  \hline
  Mean & \textbf{384.8863} & 393.6049 & 398.17 & 399.6718 & 436.2286 \\
  Standard Deviation & 56.88069 & 60.38016 & 58.63784 & 53.02499 & 67.98722  \\
   \bottomrule
\end{tabular}
\caption{Mean and standard deviation for the error rate of 10-folds cross-validation on Seoul PM2.5 dataset}
\label{tab:apps_pm}
\end{table}


\section{Conclusions}\label{sec:conclusion}

We suggested general function estimation methodologies using the B-spline basis function as the elements of an overcomplete system. The B-spline basis can systematically represent functions with varying smoothness since it has nice properties such as local support
and differentiability. The overcomplete system and a L\'{e}vy random measure enable a function that has both continuous and discontinuous parts to capture all features of the unknown regression function. Simulation studies and real data analysis also present that the proposed models show better performance than other competing models. We also showed that the prior has full support in certain Besov spaces. The prominent limitation of the LABS model is the slow mixing of the MCMC algorithm. Future work will develop an efficient algorithm for the LABS model and extend the LABS model for multivariate analysis.

\bibliographystyle{dcu}
\bibliography{larkbspline.bib}

\newpage
\appendix
\section{Proof of Theorems 1-3} \label{sec:appendA}
\subsection{Proof of Theorem \ref{thm:1}}
For simplicity, we assume that $\calX = [0,1]$. Since the B-spline basis has local support and is bounded, $\|B_k(x;\boldsymbol{\xi}_k)\|_p$ is finite for all $k \geq 0$. It is enough to show that if the Besov semi-norm, $|B_k(x; \boldsymbol{\xi}_k)|_{\mathbb{B}^\alpha_{p,q}}$ is finite for all $k \geq 0$.The definition of the modulus of smoothness and the property that $\omega_k(f, t)_p \leq 2^r \cdot \omega_{k-r}(f, t)_p,\, 0 \leq r \leq k, $ if $f \in L_p(\mathcal{X})$ imply that 
$$
\omega_r(B_k(x; \boldsymbol{\xi}_k), t)_p \leq 2^{r-1} \cdot \omega_1(B_k(x; \boldsymbol{\xi}_k), t)_p. 
$$
Let $k$ be zero. Then, the B-spline basis is piecewise constant with 2 knots, $\bxi_0 := (\bxi_{01}, \bxi_{02})$. 
By the definition of the B-spline basis \eqref{eq:bsp}, we divide into two cases to calculate the modulus of continuity. 
\begin{case}
Assume $\bxi_{01} + h < \bxi_{02}, \,\,h > 0$. Thus,
$$
\|B_0(x+h; \boldsymbol{\xi}_0) - B_0(x; \boldsymbol{\xi}_0)\|_p \leq 2 \cdot h^{\frac{1}{p}}. 
$$
\end{case}

\begin{case}
Assume $\bxi_{01} + h > \bxi_{02}, \,\,h > 0$.  Thus,
$$
\|B_0(x+h; \boldsymbol{\xi}_0) - B_0(x; \boldsymbol{\xi}_0)\|_p  \leq  2 \cdot h^{\frac{1}{p}}. 
$$
\end{case}
Therefore, in all cases, 
\begin{equation}
\omega_r(B_0(x; \boldsymbol{\xi}_0), t)_p \leq 2^{r} \cdot h^{\frac{1}{p}}.
\label{eq:lem}
\end{equation}
By definition, $|B_0(x; \boldsymbol{\xi}_0)|_{\mathbb{B}^\alpha_{p,q}} = \left(\int_0^\infty (t^{-s} \omega_r(B_0(x; \boldsymbol{\xi}_0), t)_p)^p \frac{dt}{t}\right)^{1/q}$, so
\begin{align*}
|B_0(x; \boldsymbol{\xi}_0)|_{\mathbb{B}^\alpha_{p,q}} & \leq \left[ \int_0^1 t^{-sq-1} \cdot 2^{rq} \cdot t^{\frac{1}{p}} \, dt + \int_{1}^{\infty} t^{-sq-1} \cdot 2^{rq} \, dt \right]^{1/q} \\
& = 2^r \cdot \left[ \int_0^1 t^{-q(s - \frac{1}{p}) - 1}  \, dt + \int_{1}^{\infty} t^{-sq-1}  \, dt \right]^{1/q}.
\end{align*}
The upper bound of $|B_0(x; \boldsymbol{\xi}_0)|_{\mathbb{B}^\alpha_{p,q}}$ is finite if and only if $\alpha < \frac{1}{p}$ and  $q < \infty$. 

Let $k \geq 1$. Since the B-spline basis of order $k$ is a piecewise polynomial and has $(k - 1)$ continuous derivatives at the knots, it falls in $W_p^k(\calX)$, where $W_p^k(\calX)$ is the Sobolev space, which is a vector space of functions that have weak derivatives. 
See the definition of the Sobolev space described in chapter 2.5 of \cite{devore1993constructive}. We use the following property of the modulus of smoothness,
$$
\omega_{r+k}(f, t)_p \leq t^{r} \cdot \omega_{k}(f^{(r)}, t)_p, \, t > 0, 
$$
where $f^{(r)}$ is the weak $r$th derivative of $f$. For $k \geq 1$, the Besov semi-norm of $B_k(x; \boldsymbol{\xi}_k)$ is bounded by
\begin{equation}
\begin{aligned}
|B_k(x; \boldsymbol{\xi}_k) |_{\mathbb{B}^\alpha_{p,q}} &= \left(\int_0^\infty (t^{-\alpha} \omega_r(B_k(x; \boldsymbol{\xi}_k), t)_p)^q \frac{dt}{t}\right)^{1/q}\\
&\leq \left(\int_0^1 (t^{-\alpha} \cdot (t^{k} \cdot \omega_{r-k}(B_k^{(k)}(x; \boldsymbol{\xi}_k), t)_p)^q
\frac{dt}{t} + \int_1^\infty 2^{rq} \cdot t^{-sq-1} \, dt \right)^{1/q}\\
& \leq \left(\int_0^1 (t^{-\alpha q - 1} \cdot (t^{k} \cdot 2^{r-k-1} \cdot \omega_{1}(B_k^{(k)}(x; \boldsymbol{\xi}_k), t)_p)^q
\, dt + 2^{rq} \cdot \int_1^\infty  t^{-\alpha q-1} \, dt\right)^{1/q}
\end{aligned}
\label{eq:lem1}
\end{equation}
Since $B_k^{(k)}(x; \boldsymbol{\xi}_k)$ is a piecewise constant function, \eqref{eq:lem} implies that
\begin{equation}
 \omega_{1}(B_k^{(k)}(x; \boldsymbol{\xi}_k), t)_p \leq C \cdot h^{\frac{1}{p}},\quad \text{for some constant} \,\, C > 0.
 \label{eq:lem2}
\end{equation}
Using \eqref{eq:lem1} and \eqref{eq:lem2}, it follows that
\begin{align*}
|B_k(x; \boldsymbol{\xi}_k) |_{\mathbb{B}^\alpha_{p,q}} & \leq    \left(C' \cdot  \int_0^1 t^{-\alpha q + kq + \frac{q}{p}-1} \, dt + 2^{rq} \cdot \int_1^\infty t^{-\alpha q-1}  \, dt  \right)^{1/q}  ,\quad \text{for some constant} \,\, C' > 0.
\end{align*}
For all $k \geq 1$, $|B_k(x; \boldsymbol{\xi}_k)|_{\mathbb{B}^\alpha_{p,q}}$ is finite if and only if $\alpha < k + \frac{1}{p}$ and  $q < \infty$, so the proof is complete. 

\subsection{Proof of Theorem \ref{thm:2}}

By Theorem 3 of \cite{wolpert2011stochastic}, the $L_p$ norm and Besov semi-norm of $\eta$ satisfy the following upper bounds, respectively.
\begin{equation*}
\begin{aligned}
\| \eta \|_p  & \leq  \sum_{k \in S} \sum_{l}^{} \|B_{k} (x; \boldsymbol{\xi}_{k,l})  \|_{p} |\beta_{k,l}|,  \\
 | \eta |_{\mathbb{B}^\alpha_{p,q}} & \leq \sum_{k \in S} \sum_{l}^{} | \beta_{k,l} | \cdot |B_k|_{\mathbb{B}^\alpha_{p,q}},
\end{aligned}    
\end{equation*}
 Since the condition for \eqref{eq:PMcond1} is satisfied and  B-spline basis is bounded and locally supported,  $\|\eta\|_p$ is  almost surely finite.
To obtain finite Besov semi-norms for all $k \in S$, the smoothness parameter $\alpha$  should be $\alpha < \min(S) + \frac{1}{p}$ by  \autoref{thm:bspinbsv}. Therefore, $\eta$ belongs to $\mathbb{B}^\alpha_{p,q}$ with  $\alpha < \min(S) + \frac{1}{p}$ almost surely.

\subsection{Proof of Theorem \ref{thm:3}}
For the sake of simplicity we assume $\calX = [0, 1]$. Fix $\delta > 0$ and $\eta_0 \in \mathbb{B}^{\alpha}_{p,q}([0,1])$ with $\alpha > 0, 1 \leq p, q < \infty$. If $1 \leq p' \leq p < \infty$, then $\eta_0$ also belongs to $\mathbb{B}^{\alpha}_{p',q}([0,1])$ by property of the Besov space \citep{cohen2003numerical}[3.2, page 163]. From Theorem 2.1 of \cite{petrushev1988direct}, we can show that there exists a spline $s \in S(n^{\star},q)$ such that
$$
\|\eta_0 - s  \|_{p} < C \frac{\| \eta_0 \|_{\mathbb{B}^{\alpha}_{p',q}}}{(n^{\star})^{\alpha}} < \frac{\delta}{2},
$$
where $S(n^{\star},q)$ denotes the set of all splines of degree $(q - 1)$ with a sufficiently large number $n^{\star}$ knots and constant $C = C(\alpha, p, q)$. Since any spline of given degree can be represented as a linear combination of B-spline basis functions with same degree,  we can  define a spline $s(x)$ by
\begin{equation}
s(x) = \sum_{j = 1}^{n^{\star}} \beta^*_j B^{*}_{(q-1),j}(x),  
\label{eq:bspreg}
\end{equation}
where $B^{*}_{(q-1),j}(x)$ is the B-spline basis of degree $(q-1)$ with a sequence of knots $\boldsymbol{\xi}^*$ in \eqref{eq:bsp}.

Set $n^\star := \sum_{k \in S} J_k^{\delta}$, $A := \sum_{k \in S} \sum_{l = 1}^{J_k^{\delta}} |\beta_{k,l}| < \infty$ , $\rho := \sup \| B_k (x, \boldsymbol{\xi}_{k})\|_{p} < \infty$ and  $\epsilon := \frac{\delta}{2(A + \rho)}$. We denote the range of a sequence of knots $\boldsymbol{\xi}_{k,l}$ by $r(\boldsymbol{\xi}_{k,l})$, e.g.,  $r(\boldsymbol{\xi}_{k,l}) = (\xi_{k,l,(k+2)}- \xi_{k,l,1})$.  For convenience, we reindex the coefficients and knots of the spline $s(x)$ in  \eqref{eq:bspreg} such that $\beta^*_{k,l}$ and $\boldsymbol{\xi}_{k,l}^*$ for $l = 1, \ldots, J_k^{\delta}, k \in S$. Then, the spline $s(x)$ can be expressed as follows: 
$$
s(x) = \sum_{k \in {S}} \sum_{l = 1}^{J_k^{\delta}} \beta^*_{k,l} B^{*}_{(q-1),l}(x; \boldsymbol{\xi}_{k,l}^*),
$$
where $\boldsymbol{\xi}_{k,l}^* := (\xi^*_{l},\ldots, \xi^*_{l+(q-1)+1})$ is a subsequence of given knots $\boldsymbol{\xi}^{*}$. Using the definitions of B-spline basis in \eqref{eq:bsp} and \eqref{eq:bsp2}, we can find a $\zeta > 0$ such that
$$
\max(r(\boldsymbol{\xi}_{k,l}),r(\boldsymbol{\xi}^*_{k,l})) < \zeta \Rightarrow \|B_k(x; \boldsymbol{\xi}_{k,l}) - B^{*}_{(q-1),l}(x; \boldsymbol{\xi}^*_{k,l}) \|_{p} < \epsilon, \quad \forall l,\, \forall k.
$$
 Let's define the set
$$
\bar{b'}(\eta_0) := \left\{ \eta: \eta(x) = \sum_{k \in S} \sum_{l = 1}^{J_k^{\delta}}\beta_{k,l} B_{k}(x; \boldsymbol{\xi}_{k,l}),\,\, \sum_{k \in S} \sum_{l = 1}^{J_k^{\delta}}|\beta_{k,l} - \beta^*_{k,l}| < \epsilon,\,\, 
\max(r(\boldsymbol{\xi}_{k,l}),r(\boldsymbol{\xi}^*_{k,l})) < \zeta, \, \forall l,\, \forall k\right\}.
$$

\newpage
\begin{lemma}
$$
\bar{b'}_\delta(\eta_0) \subset \bar{b}_\delta(\eta_0)
$$
\label{lem:1}
\end{lemma}

\begin{proof}
It suffices to show that $\eta \in \bar{b'}_\delta(\eta_0) \Longrightarrow \eta \in \bar{b}_\delta(\eta_0)$ to finish the proof of the lemma. For any $\eta \in \bar{b'}_\delta(\eta_0)$,
\begin{equation*}
\begin{aligned}
 \|\eta - s\|_{p} & \leq  \sum_{k \in S} \sum_{l = 1}^{J_k^{\delta}} \| \beta_{k,l} B_k(x; \boldsymbol{\xi}_{k,l}) - \beta_{k,l}^* B^{*}_{(q-1),l}(x; \boldsymbol{\xi}^*_{k,l}) \|_{p} \\
& \leq \sum_{k \in S} \sum_{l = 1}^{J_k^{\delta}} |\beta_{k,l}| \cdot \| B_k(x; \boldsymbol{\xi}_{k,l}) -  B^{*}_{(q-1),l}(x; \boldsymbol{\xi}^*_{k,l}) \|_{p} + \sum_{k \in S} \sum_{l = 1}^{J_k^{\delta}} | \beta_{k,l} - \beta_{k,l}^* | \cdot \|   B_k(x; \boldsymbol{\xi}_{k,l}) \|_{p} \\
& \leq  \epsilon \cdot \sum_{k \in S} \sum_{l = 1}^{J_k^{\delta}} |\beta_{k,l}|  + \rho \cdot \sum_{k \in S} \sum_{l = 1}^{J_k^\delta} | \beta_{k,l} - \beta_{k,l}^* |  \\
& \leq \epsilon \cdot A + 2 \rho \cdot \epsilon = (A + \rho) \cdot \epsilon = \frac{\delta}{2}.  \nonumber
\end{aligned}    
\end{equation*}
By the triangle inequality,
\begin{equation*}
\begin{aligned}
\|\eta - \eta_0  \|_{p} &\leq \| \eta - s\|_{p} + \| s - \eta_0 \|_{p}\\
& < \frac{\delta}{2} + \frac{\delta}{2} = \delta. \nonumber
\end{aligned}    
\end{equation*}
Thus, $\eta \in \bar{b}_\delta(\eta_0)$ and this finishes the proof of the lemma.
\end{proof}

\newpage
To complete the proof of this theorem, we have to show that $\Pi\left(\eta \in \bar{b'}_\delta(\eta_0)\right) > 0$ by using the previous lemma. Let $J^\star := \max_{k \in S} J_k^\delta$. By the triangle inequality,

\begin{equation*}
\begin{aligned}
\Pi\left(\eta \in \bar{b'}_\delta(\eta_0)\right) &= \Pi \left(\sum_{k \in S }\int \int_{\mathbf{R} \times \mathcal{X}^{(k+2)}} \beta_k B_k(x; \boldsymbol{\xi}_k) N_k(d \beta_k, d \boldsymbol{\xi}_k) \in  \bar{b'}_\delta(\eta_0) \right) \\
& = \Pi\left(\sum_{k \in S} \sum_{l = 1}^{J_k}B_{k}(x; \boldsymbol{\xi}_{k,l})\beta_{k,l} \in  \bar{b'}_\delta(\eta_0) \right) \\
& = \mathbb{P}\left[\sum_{k \in S}\sum_{l = 1}^{J_k^\delta} | \beta_{k,l} - \beta_{k,l}^{*}| < \epsilon,  \,\, 
\max(r(\boldsymbol{\xi}_{k,l}),r(\boldsymbol{\xi}^*_{k,l})) < \zeta,\,\,  J_k = J^\delta_k,\,\, \forall k \in S\right] \\
& > \prod_{k \in S} \left\{\mathbb{P} \left[  | \beta_{k,l} - \beta_{k,l}^{*}| < \frac{\epsilon}{|S|J^{\star}}, \max(r(\boldsymbol{\xi}_{k,l}),r(\boldsymbol{\xi}^*_{k,l})) < \zeta, \,\,l = 1,2,\ldots, J^\delta_k \right] \right\}\\
& \quad \times  \prod_{k \in S}
\left[ \dfrac{\nu_k(\mathbb{R} \times \mathcal{X}^{(k+2)})^{J_k^{\delta}}\cdot  \exp(-\nu_k(\mathbb{R} \times \mathcal{X}^{(k+2)}))}{J_k^{\delta}!} \right] \\
& = \prod_{k \in S} \left\{ \prod_{j = 1}^{J_k^\delta }\left[ \dfrac{\nu_k(|\beta_{k,l} - \beta_{k,l}^{*}| < \frac{\epsilon}{|S|J^{\star}}, \max(r(\boldsymbol{\xi}_{k,l}),r(\boldsymbol{\xi}^*_{k,l})) < \zeta)}{\nu_k(\mathbb{R} \times \mathcal{X}^{(k+2)})} \right] \right\} \\
& \quad \times \prod_{k \in S}  \left[ \dfrac{\nu_k(\mathbb{R} \times \mathcal{X}^{(k+2)})^{J_k^{\delta}}\cdot  \exp(-\nu_k(\mathbb{R} \times \mathcal{X}^{(k+2)}))}{J_k^{\delta}!} \right]  \\
& = \prod_{k \in S} \left\{ \prod_{j = 1}^{J_k^\delta }\left[ \int_{|\beta_{k,l} - \beta_{k,l}^{*}| < \frac{\epsilon}{|S|J^\star}} \pi(\beta_k) d \beta_k \int_{\max(r(\boldsymbol{\xi}_{k,l}),r(\boldsymbol{\xi}^*_{k,l})) < \zeta} \pi(\boldsymbol{\xi}_k) d \boldsymbol{\xi}_k\right] \times \left[ \dfrac{M_k^{J_k^{\delta}}\cdot  \exp(-M_k)}{J_k^{\delta}!} \right] \right\}. 
    \end{aligned}
\end{equation*}
Since we assume a finite Levy measure and $\pi(\beta_k) = \mathcal{N}(\beta_k; 0, \phi^2_{k})$, $\pi(\boldsymbol{\xi}_k) = \mathcal{U}(\mathcal{X}^{(k+2)})$ in the LABS model,
$$
\Pi\left(\eta \in \bar{b'}_\delta(\eta_0)\right) > 0.
$$
Hence, by applying the lemma, we get $\Pi\left(\eta \in \bar{b_\delta}(\eta_0)\right) \geq \Pi\left(\eta \in \bar{b'_\delta}(\eta_0)\right) > 0$  and the theorem is proved.

\newpage

\section{Full simulation results for Simulation 1} \label{sec:appendB}

This appendix contains the full simulations results of the four DJ test functions. We simulated two scenarios: (a) small sample size ($n = 128$) and (b) large sample size ($n = 512$) with different noise levels  (RSNR = 3, 5, and 10).

\begin{table}[ht!]
\centering
\adjustbox{max width=0.85\textwidth}{%

 \scriptsize
	\begin{tabular}{ccccccccc}
		\toprule
		\multirow{2}{*}{Model}&  \multicolumn{3}{c}{Bumps}& \multicolumn{3}{c}{Blocks}\\
		\cmidrule{2-7}

		& RSNR=3 & RSNR=5 & RSNR=10 & RSNR=3 &RSNR=5 & RSNR=10 \\
		\midrule

BSP-2 & 26.904(0.4461) & 25.495(0.1606) & 24.9(0.0401) & 5.96(0.4429) & 4.53(0.1594) & 3.927(0.0399) \\
  LOESS & 47.266(0.1618) & 47.163(0.0812) & 47.119(0.0366) & 17.924(0.7218) & 17.503(0.3846) & 17.332(0.2312) \\
  SS &  43.552(4.6764) & 43.377(4.7159) & 43.984(4.6074) & 4.895(0.5145) & 3.396(0.241) & 2.699(0.1107) \\
 NWK & 39.892(1.8831) & 39.365(1.3862) & 39.033(0.7393) & 4.285(0.7336) & 1.936(0.36) & 0.472(0.0567) \\
  EBW & 4.986(1.1761) & 1.936(0.601) & 0.447(0.0981) & 3.243(1.0747) & 0.859(0.2237) & 0.21(0.0484) \\
  GSP-L & 22.99(4.8847) & 22.03(5.3556) & 20.112(4.4213) & 7.409(1.3261) & 6.637(0.9807) & 6.546(1.1115) \\
  GSP-R & 41.626(2.9041) & 40.819(2.7234) & 40.955(3.102) & 15.654(2.0255) & 15.323(1.8902) & 15.44(2.3391) \\
  BPP-10 & 4.571(2.3604) & 3.668(2.7142) & 3.304(2.9789) & 2.156(0.789) & 0.908(0.2537) & 0.465(0.2403) \\
  BPP-21 & 15.115(5.9376) & 14.674(6.3396) & 14.363(6.7395) & 3.918(0.589) & 2.682(0.5399) & 2.311(0.451) \\
  BASS-1 & 2.968(0.4322) & 1.206(0.4907) & 0.252(0.0421) & 2.498(0.6331) & 0.696(0.2226) & 0.122(0.0381) \\
  BASS-2 & 47.977(7.4411) & 45.988(9.3435) & 45.021(9.2185) & 7.533(1.7616) & 4.253(0.8586) & 2.852(0.3863) \\
  LARMuK & 2.852(0.426) & 1.182(0.678) & 0.319(0.0663) & 1.799(0.5873) & 0.682(0.2436) & 0.193(0.08) \\
  LABS & \textbf{2.589(0.5908)} & \textbf{0.837(0.3124)} & \textbf{0.246(0.0683)} & \textbf{1.305(0.5272)} & \textbf{0.365(0.1645)} & \textbf{0.072(0.0293)} \\
		\bottomrule
	\end{tabular}}
	\caption{ Average mean squared errors with estimated standard errors in parentheses from 100 replications for Bumps and Blocks examples with $n = 128$}
    \label{tab:ex1-bb-128}
\end{table}

\begin{table}[ht!]
 \centering
\adjustbox{max width=0.85\textwidth}{%

 \footnotesize
	\begin{tabular}{ccccccccc}
		\toprule
		\multirow{2}{*}{Model}&  \multicolumn{3}{c}{Doppler}& \multicolumn{3}{c}{Heavisine}\\
		\cmidrule{2-7}

		& RSNR=3 & RSNR=5 & RSNR=10 & RSNR=3 &RSNR=5 & RSNR=10 \\
		\midrule

BSP-2 & 3.896(0.4928) & 2.447(0.1774) & 1.836(0.0444) & 2.399(0.4208) & 0.926(0.1515) & 0.305(0.0379) \\
  LOESS & 8.891(1.506) & 6.533(1.0853) & 5.344(0.6362) & 0.895(0.2248) & 0.548(0.0976) & 0.35(0.0436) \\
  SS &  3.644(0.587) & 2.025(0.24) & 1.251(0.0812) & 0.875(0.2725) & 0.484(0.1021) & 0.235(0.0299) \\
  NWK & 4.045(1.102) & 1.864(0.2683) & 0.477(0.0624) & 1.022(0.352) & 0.521(0.1321) & 0.228(0.0472) \\
  EBW & 2.979(0.6397) & 1.319(0.3142) & 0.341(0.0855) & 1.29(0.3473) & 0.586(0.1365) & 0.185(0.0456) \\
  GSP-L & 5.845(1.3505) & 5.313(1.3217) & 4.843(0.9023) & 2.601(2.0255) & 2.217(2.1273) & 1.743(1.7561) \\
  GSP-R & 12.164(2.9995) & 11.588(3.1093) & 12.402(3.5243) & 1.02(0.2869) & 0.756(0.1707) & 0.646(0.1152) \\
  BPP-10 & 2.911(0.6396) & 1.275(0.2471) & 0.559(0.2008) & 1.503(0.4239) & 0.674(0.1576) & 0.217(0.0616) \\
  BPP-21 & 2.575(0.5217) & 1.463(0.3399) & 1.166(0.3932) & 0.941(0.2702) & 0.444(0.105) & 0.147(0.032) \\
  BASS-1 & 2.865(0.6162) & 1.167(0.2607) & 0.353(0.0605) & 1.022(0.2434) & 0.499(0.1047) & 0.135(0.033) \\
  BASS-2 & 2.841(0.5796) & 1.753(0.2702) & 1.344(0.1205) & \textbf{0.802(0.2053)} & 0.452(0.0953) & 0.176(0.0399) \\
  LARMuK & 3(0.6708) & 1.212(0.3684) & 0.364(0.1179) & 1.13(0.3235) & 0.541(0.16) & 0.164(0.0566) \\
  LABS & \textbf{2.273(0.568)} &  \textbf{0.848(0.2521)} &  \textbf{0.234(0.0551)} &  0.897(0.242) &  \textbf{0.413(0.1492)} &  \textbf{0.103(0.0406)} \\
		\bottomrule
	\end{tabular}}
		\caption{ Average mean squared errors with estimated standard errors in parentheses from 100 replications for Doppler and Heavisine examples with $n = 128$}
    \label{tab:ex1-dh-128}
\end{table}

\begin{table}[ht!]
 \centering
\adjustbox{max width=0.85\textwidth}{%
 \footnotesize
	\begin{tabular}{ccccccccc}
		\toprule
		\multirow{2}{*}{Model}&  \multicolumn{3}{c}{Bumps}& \multicolumn{3}{c}{Blocks}\\
		\cmidrule{2-7}

		& RSNR=3 & RSNR=5 & RSNR=10 & RSNR=3 &RSNR=5 & RSNR=10 \\
		\midrule

 BSP-2 & 29.359(0.1163) & 29.011(0.0419) & 28.864(0.0105) & 5.097(0.118) & 4.74(0.0425) & 4.59(0.0106) \\
  LOESS & 43.468(5.4554) & 40.385(7.3012) & 36.495(7.2734) & 3.755(0.2529) & 3.095(0.1343) & 2.922(0.0252) \\
  SS & 16.211(0.255) & 15.406(0.123) & 15.06(0.0537) & 2.837(0.1609) & 2.197(0.0713) & 1.883(0.0256) \\
  NWK & 4.885(0.3559) & 1.796(0.1194) & 0.482(0.0319) & 2.34(0.1831) & 1.349(0.1149) & 0.581(0.1147) \\
  EBW & 2.42(0.3291) & 0.914(0.0992) & \textbf{0.272(0.0279)} & 1.227(0.2009) & 0.398(0.0753) & 0.088(0.0177) \\
  GSP-L & 16.704(2.5691) & 16.212(2.324) & 16.224(2.5594) & 3.089(0.342) & 2.613(0.3006) & 2.522(0.3358) \\
  GSP-R & 39.297(1.3464) & 39.046(1.3386) & 38.885(1.2831) & 13.847(1.3079) & 13.783(1.3156) & 13.622(1.2962) \\
  BPP-10 & 1.9(0.6277) & 1.429(0.6619) & 1.352(0.8161) & 0.493(0.137) & 0.216(0.0819) & 0.182(0.0848) \\
  BPP-21 & 4.698(0.9928) & 4.428(1.1195) & 4.308(1.1065) & 1.385(0.1974) & 0.845(0.127) & 0.665(0.1055) \\
  BASS-1 & 2.497(0.5322) & 1.679(0.4936) & 1.356(0.559) & 0.942(0.195) & 0.436(0.1189) & 0.273(0.1078) \\
  BASS-2 & 24.925(2.6356) & 23.806(2.8435) & 23.387(2.6943) & 3.033(0.2242) & 2.613(0.182) & 2.426(0.1763) \\
  LARMuK & 3.692(1.7663) & 2.465(1.0197) & 1.999(0.7813) & 1.074(0.392) & 0.646(0.2521) & 0.395(0.1903) \\
  LABS & \textbf{1.371(0.6845)} & \textbf{0.619(0.1769)} & 0.341(0.1905) & \textbf{0.363(0.1391)} & \textbf{0.113(0.0562)} & \textbf{0.021(0.009)} \\

  \bottomrule
	\end{tabular}}
	\caption{ Average mean squared errors with estimated standard errors in parentheses from 100 replications for Bumps and Blocks examples with $n = 512$}
    \label{tab:ex1-bb-512}
\end{table}

\begin{table}[ht!]
 \centering
\adjustbox{max width=0.85\textwidth}{%
 \footnotesize
	\begin{tabular}{ccccccccc}
		\toprule
		\multirow{2}{*}{Model}&  \multicolumn{3}{c}{Doppler}& \multicolumn{3}{c}{Heavisine}\\
		\cmidrule{2-7}

		& RSNR=3 & RSNR=5 & RSNR=10 & RSNR=3 &RSNR=5 & RSNR=10 \\
		\midrule

  BSP-2 & 2.875(0.1038) & 2.534(0.0374) & 2.39(0.0093) & 0.635(0.1112) & 0.294(0.04) & 0.15(0.01) \\
  LOESS & 2.756(0.3006) & 2.044(0.0474) & 1.893(0.0202) & 0.464(0.067) & 0.306(0.0523) & 0.172(0.0531) \\
  SS & 1.993(0.1234) & 1.385(0.0531) & 1.086(0.0174) & 0.384(0.0699) & 0.225(0.0299) & 0.113(0.0102) \\
  NWK & 1.649(0.1233) & 0.853(0.105) & 0.431(0.0289) & 0.398(0.0789) & 0.223(0.0327) & 0.106(0.0113) \\
  EBW & 1.263(0.1618) & 0.592(0.0809) & \textbf{0.156(0.0209)} & 0.435(0.1007) & 0.202(0.0457) & 0.066(0.0127) \\
  GSP-L & 2.167(0.2292) & 1.767(0.2376) & 1.621(0.2526) & 0.785(0.246) & 0.402(0.1555) & 0.261(0.1383) \\
  GSP-R & 9.389(1.6093) & 9.432(1.6343) & 9.1(1.577) & 0.422(0.0677) & 0.322(0.0409) & 0.279(0.0372) \\
  BPP-10 & 1.36(0.1866) & 0.632(0.0819) & 0.22(0.0283) & 0.431(0.1015) & 0.173(0.0392) & 0.055(0.0108) \\
  BPP-21 & 1.055(0.1663) & \textbf{0.496(0.0768)} & 0.249(0.0468) & 0.308(0.0771) & 0.134(0.0328) & 0.04(0.0091) \\
  BASS-1 & 1.116(0.1587) & 0.602(0.064) & 0.363(0.024) & 0.365(0.0894) & 0.149(0.0348) & 0.046(0.011) \\
  BASS-2 & \textbf{1.051(0.1916)} & 0.588(0.0804) & 0.451(0.0613) & 0.354(0.0709) & 0.169(0.0385) & 0.063(0.0131) \\
  LARMuK & 1.584(0.2401) & 1.04(0.1675) & 0.652(0.1067) & 0.413(0.1155) & 0.19(0.0543) & 0.074(0.0206) \\
  LABS & 1.243(0.2135) & 0.66(0.1141) & 0.343(0.0843) & \textbf{0.291(0.1185)} & \textbf{0.103(0.0508)} & \textbf{0.031(0.0128)} \\
  \bottomrule
	\end{tabular}}
	\caption{ Average mean squared errors with estimated standard errors in parentheses from 100 replications for Doppler and Heavisine examples with $n = 512$}
    \label{tab:ex1-dh-512}
\end{table}

\newpage
\section{Derivation of the full conditionals for LABS}  \label{sec:appendC}

In this appendix, we derive the full conditional distributions of some parameters required for Gibbs sampling. The full conditional posterior of each parameter can be easily obtained via conjugacy properties. Let us first find the full conditional posterior for $\beta_{p,q}$.
\begin{itemize}
\item Full conditional posterior for $\beta_{p,q}$

For each $q = 1,\ldots, J_p$,
{\footnotesize
\begin{align*}
[\beta_{p,q} \,|\, \beta_{p,-q}, \, \text{others},\, \mathbf{Y}] & \propto \left[ \exp \left\{-\frac{1}{2\sigma^2} \sum_{i=1}^{n}(y_i - \beta_0 - \sum_{k \in S}\sum_{l=1}^{J_k} B_{k,l}(x_i ; \boldsymbol{\xi}_{k,l}) \beta_{k,l})^2\right\} \right] \times \left[\exp \left\{ - \sum_{k \in S} \sum_{l=1}^{J_k}\dfrac{\beta_{k,l}^2}{2\sigma^2_{k}}  \right\}\right]  \\
& \propto  \exp \left\{-\frac{1}{2\sigma^2} \sum_{i=1}^{n}(y_i - \beta_0 - \sum_{k \in S \backslash \{p\}}\sum_{l=1}^{J_k} B_{k,l}(x_i ; \boldsymbol{\xi}_{k,l}) \beta_{k,l} - \sum_{l=1}^{J_p}B_{p,l}(x_i ; \boldsymbol{\xi}_{p,l}) \beta_{p,l})^2- \dfrac{1}{2\sigma^2_{p}} \sum_{k=1}^{J_p}\beta_{p,l}^2 \right\} \\
\end{align*}
} For convenience, we set $c_i = \beta_0 + \sum_{k \in S \backslash \{p\}} \sum_{l=1}^{J_k} B_{k,l}(x_i ; \boldsymbol{\xi}_{k,l}) \beta_{k,l}$ as a constant term. Then,

{\footnotesize
\begin{align*}
 & \propto  \exp \left\{-\frac{1}{2\sigma^2} \sum_{i=1}^{n}(y_i - c_i - \sum_{l=1}^{J_p} B_{p,l}(x_i ; \boldsymbol{\xi}_{p,l}) \beta_{p,l})^2 - \dfrac{1}{2\sigma^2_{p}} \sum_{l=1}^{J_p}\beta_{p,l}^2 \right\} \\
& \propto \exp \left\{-\frac{1}{2\sigma^2} \sum_{i=1}^{n}(y_i - c_i - \sum_{l \neq q}^{J_p} B_{p,l}(x_i ; \boldsymbol{\xi}_{p,l}) \beta_{p,l} - B_{p,q}(x_i ; \boldsymbol{\xi}_{p,q})\beta_{p,q})^2 - \dfrac{\beta_{p,q}^2}{2\sigma^2_{p}}  \right\} \\
& \propto \exp \left \{  -\dfrac{1}{2\sigma^2} \left( \beta^2_{p,l} \sum_{i=1}^{n}\left(B_{p,q}(x_i ; \boldsymbol{\xi}_{p,q})\right)^2 - 2\beta_{p,q} \sum_{i=1}^{n}\left(y_i - c_i - \sum_{l\neq q}^{J_p} B_{p,l}(x_i ; \boldsymbol{\xi}_{p,q})\beta_{p,l}\right) \times B_{p,q}(x_i ; \boldsymbol{\xi}_{p,q}) \right)  - \dfrac{\beta_{p,q}^2}{2\sigma^2_{p}}   \right\} \\
& = \exp \left \{ -\dfrac{1}{2} \left(\left( \dfrac{\sum_{i=1}^{n}\left(B_{p,q}(x_i ; \boldsymbol{\xi}_{p,q})\right)^2}{\sigma^2} + \dfrac{1}{\sigma^2_{p}}\right) \beta^2_{p,q}  -2 \dfrac{\sum_{i=1}^{n}\left(y_i - c_i - \sum_{l\neq q}^{J_p} B_{p,l}(x_i ; \boldsymbol{\xi}_{p,q})\beta_{p,l}\right)\left(B_{p,q}(x_i ; \boldsymbol{\xi}_{p,q})\right)}{\sigma^2} \beta_{p,q} \right) \right \}
\end{align*}
}

Thus, the full conditional distribution for $\beta_{p,q}$ is
\begin{equation*}
[\beta_{p,q} \,|\, \beta_{p,-q}, \, \text{others},\, \mathbf{Y}] \sim \mathcal{N}(\mu_{p_0}, \sigma^2_{p_0})
\end{equation*}
with
\begin{gather*}
  \sigma^2_{p_0} = \left( \dfrac{\sum_{i=1}^{n}\left(B_{p,q}(x_i ; \boldsymbol{\xi}_{p,q})\right)^2}{\sigma^2} + \dfrac{1}{\sigma^2_{p}}\right)^{-1} \\
   \mu_{p_0} = \sigma^{2}_{p_0}  \times \dfrac{\sum_{i=1}^{n}\left(y_i - c_i - \sum_{l\neq q}^{J_p} B_{p,l}(x_i ; \boldsymbol{\xi}_{p,q})\beta_{p,l}\right)\left(B_{p,q}(x_i ; \boldsymbol{\xi}_{p,q})\right)}{\sigma^2}.
\end{gather*}

\bigskip
\item Full conditional posterior of $M_k$

For each $k \in S$,
\begin{align*}
[M_k \,|\, \text{others}] & \propto M_k^{J_k} \exp \{-M_k\}\times M_k^{a_{\gamma_k} -1} \exp \{-b_{\gamma_k} M_k\} \\
& = M^{J_k + a_{\gamma_k} -1} \exp \{ -(1+ b_{\gamma_k})M_k\}
\end{align*}

\bigskip

The full conditional distribution for $M_k$ is given by

\begin{equation*}
[M_k \,|\, \text{others}]  \sim \text{Ga}(a_{k}, b_{k})
\end{equation*}
where
\begin{align*}
& a_{k} = a_{\gamma_k} + J_k,\\
& b_{k}= b_{\gamma_k} + 1.
\end{align*}

\item Full conditional posterior of $\sigma^2$

\begin{align*}
[\sigma^2 \,|\, \text{others}, \mathbf{Y}] & \propto \left[(\sigma^2)^{-\frac{n}{2}}\times \exp \left\{-\frac{1}{2\sigma^2} \sum_{i=1}^{n}(y_i - \beta_0 - \sum_{k \in S}\sum_{l=1}^{J_k} B_{k}(x ; \boldsymbol{\xi}_{k,l}) \beta_{k,l})^2\right\} \right] \\
& \times \left[(\sigma^2)^{-\frac{r}{2}+1}\times \exp \left\{-\frac{rR}{2\sigma^2}\right\} \right] \\
& \propto (\sigma^2)^{-\frac{n+r}{2}+1}\exp \left\{ -\dfrac{1}{2 \sigma^2}\sum_{i=1}^{n}(y_i - \beta_0 - \sum_{k \in S}\sum_{l=1}^{J_k} B_{k}(x ; \boldsymbol{\xi}_{k,l}) \beta_{k,l})^2 - \dfrac{rR}{2\sigma^2} \right\}
\end{align*}

\bigskip

The full conditional distribution for $\sigma^2$ is
\begin{equation*}
[\sigma^2 \,|\, \text{others}, \mathbf{Y}] \sim \text{IG}\left(\dfrac{r_0}{2}, \dfrac{r_0R_0}{2} \right)
\end{equation*}
with
\begin{gather*}
 r_0 = r+n,\\
 R_0 = \dfrac{\sum_{i=1}^{n}(y_i - \beta_0 - \sum_{k \in S}\sum_{l=1}^{J_k} B_{k}(x ; \boldsymbol{\xi}_{k,l}) \beta_{k,l})^2 + rR}{r_0}.
\end{gather*}

\end{itemize}

\end{document}